\documentclass[11pt]{article}

\usepackage[margin=0.8in]{geometry}

\usepackage{algorithm}
\usepackage[noend]{algpseudocode}
\usepackage{tikz}
\usepackage{caption}

\usepackage{amsmath}
\usepackage{amsfonts}
\usepackage{amssymb}
\usepackage{amsbsy}
\usepackage{enumitem}

\usepackage[pdftex,colorlinks=true,urlcolor=blue,citecolor=black,anchorcolor=black,linkcolor=black]{hyperref}
\usepackage{cleveref}
\crefname{assumption}{assumption}{assumptions}

\crefdefaultlabelformat{#2\textup{#1}#3}
\Crefname{assumption}{Assumption}{Assumptions}
\Crefrangeformat{assumption}{Assumptions #3\textup{#1}#4--#5\textup{#2}#6}
\crefname{secinapp}{appendix}{appendices}
\Crefname{secinapp}{Appendix}{Appendices}

\usepackage{nccmath} 
\usepackage{mathtools,tabularx}
\newcommand{\vect}[1]{\boldsymbol{#1}} 
\newcommand{\matr}[1]{\mathbf{#1}}

\usepackage{multirow}
\newcommand{\E}{\mathrm{E}}
\newcommand{\V}{\mathrm{Var}}

\newcommand{\N}{\mathrm{N}}
\newcommand{\Cov}{\mathrm{Cov}}
\newcommand{\Etilde}{\widetilde{\E}}
\newcommand{\Vtilde}{\widetilde{\V}}
\newcommand{\Covtilde}{\widetilde{\Cov}}
\newcommand{\Fhat}{\widehat{F}}
\newcommand{\bFhat}{\widehat{\vect{F}}}
\newcommand{\bF}{\vect{F}}
\newcommand{\one}[1]{\delta\{#1\}}
\newcommand{\set}[1]{\{#1\}}
\newcommand{\bX}{\vect{X}}
\newcommand{\bXdata}{\matr{X}_p}
\newcommand{\bx}{\vect{x}}
\newcommand{\bZ}{\vect{Z}}
\newcommand{\bbZ}{\matr{Z}}
\newcommand{\IF}{\mathrm{IF}}
\newcommand{\IFipx}{\IF_{ip}(\bx; \bF^c)}
\newcommand{\IFlpx}{\IF_{\ell p}(\bx; \bF^c)}
\newcommand{\IFllpx}{\IF_{\ell'p}(\bx; \bF^c)}
\newcommand{\IFhat}{\IF}
\newcommand{\IFhathat}{\widehat{\IF}}
\newcommand{\IFhatXipj}{\IFhat_{ip}(\bX_{pj}; \bFhat)}
\newcommand{\IFhatXlpj}{\IFhat_{\ell p}(\bX_{pj}; \bFhat)}

\newcommand{\IFXipj}{\IF_{ip}(\bX_{pj}; \bF^c)}
\newcommand{\IFXlpj}{\IF_{\ell p}(\bX_{pj}; \bF^c)}

\newcommand{\IFhathatXipj}{\IFhathat_{ip}(\bX_{pj}; \bFhat)}
\newcommand{\IFhathatXlpj}{\IFhathat_{\ell p}(\bX_{pj}; \bFhat)}

\newcommand{\lambst}{\vect{\lambda}^*}
\newcommand{\norm}[1]{\|#1\|}
\newcommand{\bignorm}[1]{\bigg\|#1\bigg\|}

\newcommand{\bw}{\vect{w}}
\newcommand{\bG}{\vect{G}}

\newcommand{\bmu}{\boldsymbol{\mu}}

\newcommand{\argmax}{\mathrm{argmax}}

\usepackage{natbib}
 \bibpunct[, ]{(}{)}{,}{a}{}{,}%

\newtheorem{theorem}{Theorem}

\newtheorem{proposition}{Proposition}

\newtheorem{corollary}{Corollary}

\newtheorem{assumption}{Assumption}

\newtheorem{lemma}{Lemma}

\newenvironment{proof}[1]{%
  \par\noindent\textbf{#1. }\ignorespaces
}{%
  \hfill$\square$\par
}

\begin{document}

\title{Nonparametric Robust Comparison of Solutions under Input Uncertainty\thanks{This work is supported by the National Science Foundation under Grant Nos.\ CAREER CMMI-2045400,\\ CMMI-2417616, and DMS-2410944.}}

\author{Jaime Gonzalez-Hodar\thanks{H.~Milton Stewart School of Industrial and Systems Engineering, Georgia Institute of Technology.} \and Johannes Milz\footnotemark[2] \and Eunhye Song\footnotemark[2]}

\maketitle

\abstract{
We study ranking and selection under input uncertainty in settings where additional data cannot be collected.  We propose the Nonparametric Input-Output Uncertainty Comparisons (NIOU-C) procedure to construct a confidence set that includes the optimal solution with a user-specified probability. We construct an ambiguity set of input distributions using empirical likelihood and approximate the mean performance of each solution using a linear functional representation of the input distributions. By solving optimization problems evaluating worst-case pairwise mean differences within the ambiguity set, we build a confidence set of solutions indistinguishable from the optimum. We characterize sample size requirements for NIOU-C to achieve the asymptotic validity under mild conditions. Moreover, we propose an extension to NIOU-C, NIOU-C:E, that mitigates conservatism and yields a smaller confidence set. In numerical experiments, NIOU-C provides a smaller confidence set that includes the optimum more frequently than a parametric procedure that takes advantage of the parametric distribution families. 
}%



\section{Introduction}
\label{sec:intro}

Stochastic simulation is widely used for decision-making owing to its flexibility and ability to model randomness. Simulation outputs depend on input random vectors from models estimated using data. When data collection is costly or limited to a single batch, input models are subject to finite-sample error, introducing extra uncertainty beyond run-to-run variability. This additional uncertainty, called \emph{input uncertainty} (IU)~\citep{barton2022}, cannot be reduced when further data collection is infeasible.

We consider ranking and selection (R\&S)   where the goal is to find the solution with the largest expected system output, represented  by the simulation output mean. 
This setting arises  when additional input data collection is infeasible or far costlier than simulation replications. For instance, 
General Motors uses simulation to compare  vehicle content portfolios, with  input models estimated from one-time consumer survey data~\citep{GMVCO}. 
Identifying the best solution  is challenging when the performance measures differences  are  marginal~\citep{songnelsonhong2015}.
Throughout the paper, we assume that the simulated optimum under the true distributions is the system optimum.

Our goal is to infer the solution by constructing a confidence set of solutions that contains the solution with a user-specified probability, accounting for IU and simulation error in the performance measure estimates, and ensuring asymptotic validity as sample sizes grow. 
A similar problem is studied in~\cite{Song2019}, which introduces the Input–Output Uncertainty Comparisons (IOU-C) and assumes the input distributions have known parametric distribution functions. 
This assumption simplifies input modeling to parameter estimation and enables parametric response models (e.g., regression) to link input and output. Consequently, quantifying IU reduces to characterizing the sampling distribution of the parameter estimator and fitting the response model.

However, parametric assumptions on the input distribution are often unrealistic. Without strong physical justification,  the parametric input model is typically chosen by fitting a number of candidate distributions to the data and picking the best goodness-of-fit test score. 
Misspecification can severely impact decisions. For instance, \cite{songlambarton2024} empirically show that fitting an exponential distribution to bimodal service times in a queueing  model causes severe undercoverage of  confidence intervals (CIs) for the  mean performance measure.  
Moreover, bias from an incorrect parametric model persists even with a large input data size. 

We  propose a nonparametric framework to construct the confidence set directly from the input data.
Using empirical likelihood (EL)~\citep{owenEL}, we  quantify the uncertainty about the true input distributions without parametric assumptions.  
EL  enables constructing an ambiguity set of input models, which asymptotically contains a distribution function yielding the correct simulation output means for all $k$ solutions with a given probability as input size grows. Leveraging this result, we adapt the  classical multiple comparisons with the best (MCB) procedure~\citep{chang1992optimal} by computing the confidence region of simulation output mean differences through $k(k-1)$ optimization problems, each maximizing the pairwise mean difference subject to the ambiguity set. 

Solving these optimization problems is challenging as the objectives must express the simulation output mean as a functional of each input model in the ambiguity set.  Therefore, we approximate the objectives using first-order functional Taylor expansions, where dependence on the input models is represented by influence functions~\citep{hampel1974influence} estimated via simulations. This yields convex problems, which can be solved to arbitrary precisions at a low computational cost. The CI bounds are then estimated by running simulations at the maximizers. We formalize this procedure as an algorithmic framework, Nonparametric IOU-C (NIOU-C). We show that NIOU-C returns the confidence set including the optimum with asymptotically correct coverage probability  provided that simulation sample sizes grow faster than the input data size. In numerical studies, NIOU-C delivers the correct empirical confidence with a smaller confidence set than an IOU-C algorithm in all experiments. Notably, even if the algorithm correctly assumes the parametric input distributions,  NIOU-C still outperforms it.

\subsection{Literature Review}
\label{subsec:lit.review}

Our work extends and builds on a number of research fields, which we briefly review below.

\paragraph{Nonparametric IU quantification}
Prior work on nonparametric IU quantification mainly constructs confidence intervals for a single simulation output mean; see Section~17.3.2 of \citet{barton2022} for a  review.   Most approaches adopt bootstrap to approximate the input data's sampling distribution~\citep{songlambarton2024}. 
Our work is  closest to \citet{lamqian2016,lam2017optimization}, which also use EL to construct CIs accounting for IU. 
However, their inference is one-dimensional, while we extend EL to multiple dimensions for joint inference on all R\&S solutions.

\paragraph{Simulation optimization under IU}

Existing work either consider fixed input batches or allow for additional data collection. 
A detailed discussion on the latter can be found in~\citet{hesong2024}. 
For fixed batches, \cite{corlu2013subset, corlu2015subset, Song2019} study the confidence set inference  for parametric input models.
\cite{pearce2017bayesian,wang2020} instead optimize  expected simulation output accounting for simulation error and IU via Bayesian optimization, 
representing risk-neutral approaches. On the other hand, \citet{gao2017robust,Fan2020} consider  robust formulations modeling IU with a finite set of candidate distributions. They aim to find the solution with the best worst-case performance and propose sequential sampling algorithms. \citet{wu2018} study R\&S  with risk measures (e.g.\ quantiles) under  parametric Bayesian  IU without providing specific solution algorithms. These frameworks 
assume known parametric families of  input models, or restrict  candidate distributions to finite sets.

\paragraph{Distributionally robust optimization (DRO)}

In DRO, ambiguity sets represent distributional uncertainty in stochastic optimization. \citet{delage2010} use moment-based ambiguity sets, \citet{BenTal2013} employ $f$-divergence ball, and
\citet{esfahanikuhn2018} adopt  Wasserstein balls for  high-probability coverage.
\citet{lam2019} further links EL-based ambiguity sets to Burg entropy. While we also use ambiguity sets, our focus differs: rather than optimizing decisions under worst-case distributions, we evaluate worst-case pairwise performance differences within the ambiguity set to form a confidence region for solution means, from which we construct the confidence set.

\subsection{Contribution Summary}

We propose a data-driven procedure  for   R\&S under IU that provides a confidence set containing the unknown best with a user-specified probability. 
Our main contributions are fourfold:
\textbf{(i)}~We provide the first nonparametric framework for R\&S under IU,  improving real-world applicability.
\textbf{(ii)}~Using EL, we  model uncertainty in solution comparisons via an ambiguity set, and approximate simulation-based performance measures using influence functions. 
To build the confidence set, we then formulate optimization problems solvable with standard solvers.
\textbf{(iii)} We prove that the confidence set provides an asymptotically valid coverage when simulation sample sizes grow faster than the input data size regardless of the number of data sources or the dimensionality of the input data.
\textbf{(iv)}  We benchmark  against IOU-C on examples where it either  correctly specifies or misspecifies the distribution families.  Our confidence set achieves correct coverage  with the same simulation budget while IOU-C fails 
under misspecification and becomes conservative otherwise.

A conference version of this work appeared in \citet{gonzalez2024},
which we extend in several directions. 
First, we establish a rigorous theoretical foundation for NIOU-C and prove that the proposed procedure achieves the desired coverage guarantee in the limit. 
Second, we propose NIOU-C:E, an efficiency-oriented variant designed to reduce the conservatism in  confidence sets. 
Third, we expand the empirical study by  conducting a comprehensive sensitivity analysis to examine how approximation errors affect the performance of NIOU-C and NIOU-C:E under varying input sizes, simulation budgets, and the numbers of  solution.

\subsection{Paper Outline}
In \Cref{sec:problem_desc},  we describe the R\&S problem under IU. 
\Cref{sec:as} formulates the ambiguity set  using EL\@. 
\Cref{sec:model}  presents a linear approximation of the simulation output mean  via the influence function. 
\Cref{sec:optimization} formulates optimization problems for a confidence region for the performance measure differences  and their solution methods are summarized in \Cref{sec:as-algorithm}. 
\Cref{sec:extension} explores an alternative  ambiguity set for tighter coverage. 
\Cref{sec:empirical}  provides empirical validation. 
Finally, we conclude in \Cref{sec:conclusion}. All  proofs of the theoretical statements are in the Online Supplement.

\section{Problem Description}
\label{sec:problem_desc}

We consider an R\&S problem with $k \in \mathbb{N}$ candidate solutions, evaluated by a stochastic simulator.
The system's randomness is modeled by a collection of input distributions, $\bF^c \triangleq \{F_p^c\}_{p \in [m]}$, 
where each $F_p^c$ is a probability density function
defined on a subset $\Xi_p$ of Euclidean space, and $[m]\triangleq\{1,\ldots,m\}$. We assume that the input distributions are mutually independent, meaning that the random inputs \( \mathbf{X}_1, \ldots, \mathbf{X}_m \), with \( \mathbf{X}_p \sim F_p^c \), are statistically independent across \( p \in [m] \).

For each Solution $i$, the simulation output $Y_i(\bF)$ depends on an arbitrary input model $\bF\triangleq \{F_p\}_{p\in[m]}$.  A replication of $Y_i(\cdot)$ requires $T_{ip}$ independent and identically distributed (i.i.d.) realizations, $\bZ^i_{p}(1),\ldots,\bZ^i_{p}(T_{ip}) \sim F_p$, denoted by $\matr{Z}^i_{p} \triangleq \{\bZ^i_{p}(t)\}_{t\in [T_{ip}]}$ from each $p\in[m]$. We formalize the input-output relationship and a bounded moment condition on $h_i$ below. 

\begin{assumption} \label{assump:input.length}
For each $i\in[k]$, $Y_{i}(\bF) = h_i(\matr{Z}^i_{1},\ldots,\matr{Z}^i_{m})$, where $h_i \colon 
\Xi_1^{T_{i1}} \times \cdots \times 
\Xi_m^{T_{im}}
\to \mathbb{R}$ is the output function of Solution $i$ and $0<T_{ip}<\infty$ is a known constant for each $p \in [m]$. Furthermore, for each $i\in[k]$, $h_i$
is measurable and satisfies
$
\int |h_i(\boldsymbol{x}_1, \ldots, \boldsymbol{x}_m)|^2 \, 
    \prod\nolimits_{p=1}^m  \prod\nolimits_{t=1}^{T_{ip}}
    \mathrm{d} F_p^c(\boldsymbol{x}_{p,t})
< \infty
$.
\end{assumption}

In general, $h_i$ does not have an analytical expression. For instance, $h_i$ may be an abstract representation of how inputs are mapped to the summary statistic computed at the end of a replication of a discrete-event simulator. While we adopt both $Y_i(\bF)$ and $h_i(\matr{Z}^i_{1},\ldots,\matr{Z}^i_{m})$ to denote the output, we use the latter whenever we need to explicitly represent the dependence on the inputs. 
We define the performance measure of Solution $i$, $\eta_i(\bF)$, as the expectation of the simulation output, which can explicitly be rewritten in an integral form under \Cref{assump:input.length}: 
\begin{equation*}
    \eta_i(\bF) \triangleq \E[Y_i(\bF)]
    =
    \int h_i(\boldsymbol{x}_1, \ldots, \boldsymbol{x}_m) \, 
    \prod\nolimits_{p=1}^m  \prod\nolimits_{t=1}^{T_{ip}}
    \mathrm{d} F_p(\boldsymbol{x}_{p, t}),
    \quad 
    \text{where}
    \quad 
    \boldsymbol{x}_p = \{\boldsymbol{x}_{p,t}\}_{t \in [T_{ip}]}.
\end{equation*}
The goal 
is to find the solution with the largest performance measure when $\bF^c$ is the input model: 
\begin{equation}\label{eq:RnS}
    i^c \triangleq \arg\max\nolimits_{i\in [k]} \eta_i(\bF^c).
\end{equation}
We assume $i^c$ to be unique in this work, i.e., $\eta_{i^c}(\bF^c) \neq \eta_i(\bF^c)$ for all $ i \in [k] \setminus \{i^c\}$.

The true input model $\bF^c$ is typically unknown, but may be estimated from data. For each $p \in [m]$, we assume access to a fixed 
batch of $n_p$ i.i.d.\ samples, $\{\bX_{pj}\}_{j \in [n_p]}$, from $F_p^c$, collected independently across $p\in[m]$. 
Let 
$n \triangleq (1/m)\sum_{p=1}^m n_p$ denote the average sample size from $m$ sources.
Although our methodology is specifically designed for problems where additional data collection is infeasible,  in later sections we analyze the  asymptotic behavior as $n \to \infty $ to assess the precision of statistical inference under increasing data availability.
To this end, we make the 
following assumption. 
\begin{assumption}\label{assump:data_size}
    There exists $\beta_p > 0$ such that $n_p / n \to \beta_p$ as $n \to \infty$.
\end{assumption}

While one could plug in the empirical estimates of $\bF^c$ to solve \eqref{eq:RnS}, this may not reliably identify the  correct solution $i^c$ with a high probability~\citep{songnelsonhong2015}. Instead, we take an inference approach, constructing a data-driven, computationally tractable confidence set $I_n \subset [k]$ such that
\begin{equation}
\label{eq:asymptotic-coverage}
    \liminf\nolimits_{n \to \infty} \Pr\{i^c \in I_n\} \geq 1 - \alpha.
\end{equation}
To construct $I_n$, we first establish the following lemma, generalizing Theorem~1.1 in \cite{chang1992optimal} to construct MCB CIs without requiring point estimates of $\eta_i(\bF^c)$, and to allow for asymptotic coverage guarantees.
For $x \in \mathbb{R}$, let $x^+ \triangleq \max\set{0,x}$ and $x^- \triangleq \max\set{0,-x}$.

\begin{lemma}[Asymptotic MCB-CI]
\label{lem:asymptotic-mcb-ci}
    Let $\alpha \in (0,1)$.
    If for each $i\in [k]$ there exists a sequence of random variables 
    $\{U_{i\ell,n}\}_{n\ge 1, \ell \neq i}$ such that
    \begin{equation} \label{eq:asymptotic-mcb-ci-1}
    \liminf_{n\to\infty} \Pr \big\{ \eta_i(\bF^c) - \eta_\ell(\bF^c) \leq 
    U_{i\ell,n}, \forall \ell \neq i \big\} \geq 1 - \alpha,
    \end{equation}
    then
    \begin{equation} \label{eq:asymptotic-mcb-ci-2} 
    \liminf_{n\to\infty} \Pr \big\{ \eta_i(\bF^c) - 
    \max\nolimits_{\ell  \neq i} {\eta_\ell(\bF^c)} \in 
    [D_{i,n}^-, D_{i,n}^+], \forall  i \in [k] \big\} \geq 1 - \alpha,
    \end{equation} 
    where $D_{i,n}^+ \triangleq \big(\min_{\ell  \neq i}\set{U_{i\ell,n}}\big)^+$, 
    $I_n \triangleq \set{i \in [k]: D_{i,n}^+ > 0}$, and 
    $D_{i,n}^- = 0$ if $I_n = \set{i}$, and 
    $D_{i,n}^- = -\big(\min_{\ell  \in I_n: \ell \neq i}\set{-U_{\ell i,n}}\big)^-$, 
    otherwise.
    Moreover, 
    $\liminf_{n\to\infty}\Pr\{i^c\in I_n\}\geq 1-\alpha$.
\end{lemma}

\Cref{lem:asymptotic-mcb-ci} states that if, for each $i \in [k]$, we construct  simultaneous $1-\alpha$ CIs for the $k-1$ pairwise differences $\{\eta_i(\bF^c) - \eta_\ell(\bF^c)\}_{\ell\neq i}$, then we can construct MCB CIs that cover the difference between each solution's mean and the best of the rest with the same probability. 
To obtain the bounds $\{U_{i\ell,n}\}_{\ell \neq i}$  satisfying \eqref{eq:asymptotic-mcb-ci-1}, we proceed in three steps.
(i)~We define a data-driven ambiguity set $\mathcal{U}_\alpha$
to capture uncertainty in $\bF^c$ using EL\@. 
(ii)~We introduce a tractable linear model $\eta_i^L(\bF)$ of $\eta_i(\bF)$ for $\bF \approx \bF^c$.
(iii)~We construct a confidence region for $\{\eta_i(\bF^c) - \eta_\ell(\bF^c)\}_{\ell\neq i}$ for each $i\in[k]$ by solving  worst-case optimization problems with $\{\eta_i^L(\bF) - \eta_\ell^L(\bF)\}_{\ell\neq i}$ as objectives subject to $\bF\in \mathcal{U}_\alpha$ and derive $I_n$ from the maximizers via \Cref{lem:asymptotic-mcb-ci}.

\section{Ambiguity Set and Theoretical Coverage Guarantee}
\label{sec:as}

In this section, we construct an ambiguity set for a generic random vector $\vect{Y}_p \in \mathbb{R}^{q}$ with variance-covariance matrix $\matr{V}_p$ for $p\in[m]$. Then, in \Cref{sec:optimization}, we establish the confidence set in \Cref{lem:asymptotic-mcb-ci}.

For each $p\in[m]$, let $\vect{Y}_{pj}$, $j \in [n_p]$ be i.i.d.\ samples of $\vect{Y}_p$. For $\vect{a}, \vect{b}\in\mathbb{R}^q$, we write $\vect{a} \geq \vect{b}$ if $a_i \geq b_i$, $i \in [q]$. We construct a confidence region containing  $\bar{\bmu} \triangleq \sum_{p=1}^m \E[\vect{Y}_p]$ with probability at least $1-\alpha$ by finding  $\vect{B} \triangleq \{B_\ell\}_{\ell\in[q]} \in \mathbb{R}^q$ such that $\Pr\{\bar{\bmu} \leq \vect{B}\} \geq 1-\alpha$.
For $\bmu\in\mathbb{R}^q$, we define 
\begin{equation} \label{eq:R}
    R(\bmu) \triangleq \max\Big\{\, \prod\nolimits_{p=1}^m\prod\nolimits_{j=1}^{n_p} n_p {w_{pj}} \colon \sum\nolimits_{p=1}^m \sum\nolimits_{j=1}^{n_p} \vect{Y}_{pj} w_{pj} = \bmu, \quad \vect{w}_p \in \Delta_{n_p}, \quad  p \in [m] \, 
    \Big\}.
\end{equation}
Here, $\prod_{p=1}^m\prod_{j=1}^{n_p} n_p {w_{pj}}$ is the ratio between a nonparametric likelihood and the maximum likelihood under uniform weights. 
Therefore, 
$R(\bmu)$ is the maximum EL over the set of empirical distributions $\bw$ such that the sum of the means of  $\vect{Y}_p, p\in[m],$ stipulated by $\bw$, equals $\bmu$.
The next result provides the asymptotic theory to infer the mean of the sum of multi-dimensional vectors.

\begin{theorem}
    \label{thm:multisample-el} 
    Suppose that $\vect{Y}_{p}\in\mathbb{R}^q$, $p\in[m],$ are mutually independent with  positive definite variance-covariance matrix $\matr{V}_p$. Moreover, for each $p\in[m]$, $n_p/n \to \beta_p > 0$ as $n \to \infty$, where $n = (1/m)\sum_{p=1}^m n_p$.
    Then, $- 2\log(R(
    \bar{\bmu})) \xrightarrow{D}
    \chi_{q}^2$
    as $n \to \infty$.
\end{theorem}

While EL generally permits  working with rank-deficient variance-covariance matrices (see \cite{Owen1988}), the limiting EL distribution’s degrees of freedom depend on the (generally unknown) rank. To simplify analysis, we work full rank matrices.
\Cref{thm:multisample-el} extends existing results. 
For $m=1$, it recovers  Theorem~3.2 in~\cite{owenEL}, which focuses on the mean inference for a random vector. For $q = 1$, 
\Cref{thm:multisample-el} corresponds
to Theorem~4 in \cite{lam2017optimization}, which studies the mean inference for sums of random variables.

\Cref{thm:multisample-el} guarantees the event $\{-2\log R(\bar{\bmu})\leq \chi_{q,1-\alpha}^2\}$ holds with asymptotic probability $1-\alpha$. Since $\log$ is monotone, 
$\log R(\bar{\bmu}) = \max\big\{\sum_{p=1}^m\sum_{j=1}^{n_p} \log(n_p {w_{pj}}) \colon \sum_{p=1}^m \sum_{j=1}^{n_p} \vect{Y}_{pj} w_{pj} = \bar{\bmu}, \vect{w}_p \in \Delta_{n_p},  p \in [m]\big\}$. 
This yields the data-driven ambiguity set
\begin{equation}
\label{eq:uncertaintyset}
    \mathcal{U}_\alpha \triangleq \Big\{ \bw \in \mathbb{R}^{n_1 + \cdots + n_m} 
    \colon -2 \sum\nolimits_{p=1}^m \sum\nolimits_{j=1}^{n_p} \log(n_p w_{pj}) 
    \leq \chi^2_{q,1-\alpha},\; \bw_p \in \Delta_{n_p},\; p \in [m] \Big\},
\end{equation}
where $\Delta_s \triangleq \{ \boldsymbol{v} \in \mathbb{R}^s \colon 
\boldsymbol{v} \geq 0, \sum_j v_j = 1 \}$ is the  
probability simplex. Each $\bw \in \mathcal{U}_\alpha$ defines a reweighted empirical distribution
$\bF$, where $F_p(\bx_p) \triangleq 
\sum_{j=1}^{n_p} w_{pj} \one{\bX_{pj} = \bx_p}$ and $\delta\{\cdot\}$ is an indicator function.
Let $\bFhat \triangleq (\Fhat_1,\ldots,\Fhat_m)$ 
be the empirical distribution function (edf) 
constructed from the input data. 

\begin{figure}[tbp]
    \centering
    \includegraphics[width=0.4\textwidth]{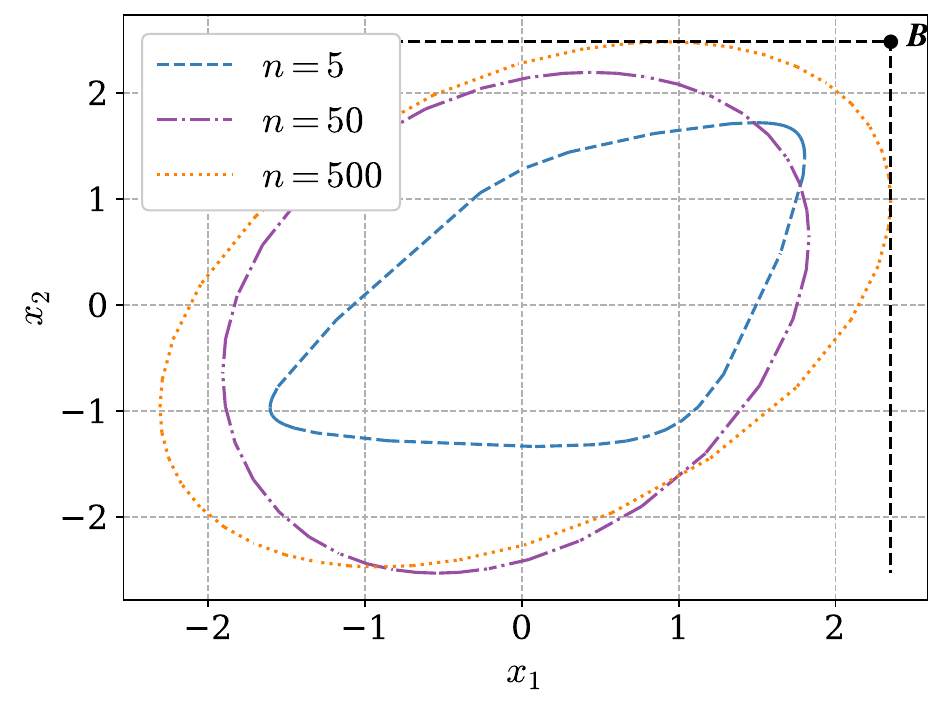}
\caption{
Illustration of the confidence set $C_\alpha$ given in \eqref{eq:confidence-set} for sample sizes $n=5, 50, 500$ and a polyhedral superset of $C_\alpha$
defined by the optimal values $\vect{B} = \{B_\ell\}_{\ell \in [2]}$ in \eqref{eq:generic_opt}
for $n=500$,
with $\alpha=0.05$, $q=2$, $m=1$ and $\vect{Y}_{\cdot j}$  normally distributed with zero mean, unit variance and  correlation equal to $0.4$.
\label{fig:ellipsoid_manybox}
Each plot of the boundary of $C_\alpha$ is scaled and translated to facilitate comparison of their shapes. The boundary of $C_\alpha$ is generated by plotting $\argmax_{\vect{c} \in C_\alpha} \vect{a}^\top \vect{c}$ for every possible direction $\vect{a} \in \mathbb{R}^2$.
}
\end{figure}

Compared to \eqref{eq:R}, \eqref{eq:uncertaintyset} drops the maximization, implying that $\mathcal{U}_\alpha$ contains weights $\vect{w}$ for which the weighted average of $\{\vect{Y}_{pj}\}_{j\in[n_p]}$ equals $\bar{\bmu}$ with an asymptotic probability of at least $1-\alpha$.
This leads to a confidence region $C_\alpha$, defined as all mean values attainable by distributions in $\mathcal{U}_\alpha$,
\begin{equation}
\label{eq:confidence-set}
C_\alpha \triangleq \Big\{ \sum\nolimits_{p=1}^m \sum\nolimits_{j=1}^{n_p} 
\vect{Y}_{pj} w_{pj} \colon \bw \in \mathcal{U}_\alpha \Big\}.
\end{equation}
\Cref{thm:multisample-el} ensures that $\Pr\{\bar{\bmu} \in C_\alpha\}$ approaches $1-\alpha$ as the sample size tends to infinity.
To compute $B_\ell$ for each $\ell \in [q]$, we maximize the vectors in $C_\alpha$ along the $\ell$th dimension:
\begin{equation}\label{eq:generic_opt}
    {B}_\ell 
    \triangleq \max\nolimits_{\vect{c} \in C_\alpha} \vect{e}_\ell^\top \vect{c},
\end{equation}
where 
$\vect{e}_\ell \in \mathbb{R}^q$ is the $\ell$th standard basis vector in $\mathbb{R}^q$. 
This yields $C_\alpha \subset \{\boldsymbol{\mu} \in \mathbb{R}^q \colon 
\boldsymbol{\mu} \leq \vect{B}\}$. 
\Cref{fig:ellipsoid_manybox} depicts  $C_\alpha$ for different sample sizes and  $\vect{B}$ for the $n=500$ case.

Our next result shows that the bounds $\vect{B}$ provide the correct coverage as the data size grows. 
\begin{lemma}
    \label{lem:mean.CR.coverage}
    Under the hypotheses of \Cref{thm:multisample-el}, 
    $\liminf_{n \to \infty} \Pr\{\bar{\bmu} \leq \vect{B}\} \geq 1 - \alpha$.
\end{lemma}

\section{Linear Approximation of Performance Measures}
\label{sec:model}

In this section, we develop a linear approximation of the performance measures $\eta_i$ 
(see, e.g., \cite[Section~20.1]{Vaart1998}).
Under some technical assumptions stated below, we can
approximate $\eta_i$ 
at $\bF = (F_1, \ldots, F_p)$ 
via the linear functional in $\bG$:
\begin{equation} \label{eq:nonparam.delta}
     \eta_i^L(\bG; \bF) \triangleq \eta_i(\bF) + \sum\nolimits_{p=1}^m \int \IF_{ip}(\bx_p;\bF) \, \mathrm{d}\bG_p(\bx_p),
\end{equation} 
where $\IF_{ip}(\bx_p;\bF)$ is the influence function of $\eta_i$, that is, 
\begin{equation} \label{eq:influence_function_derivative}
\IF_{ip}(\bx_p;\bF) = \lim_{\epsilon \to 0^+} \frac{\eta_i(F_1, \ldots, F_{p-1}, F_{p} + \epsilon (\delta_{\bx_p}-F_p), F_{p+1}, \ldots, F_{m}) - \eta_i(F_1,\ldots,F_m)}{\epsilon},
\end{equation}
where $\delta_{\bx_p}$ puts the unit probability on $\bx_p$. 
The expansion \eqref{eq:nonparam.delta} holds under certain conditions on $\bG$, $\eta_i$, and $\bF$ which are discussed below.
The expansion in~\eqref{eq:nonparam.delta} resembles a parametric first-order $m$-variate Taylor's expansion, where the influence function, $\IFipx$, acts as a functional derivative and the integral captures the directional change in $\eta_i$ owing to perturbations in the input distribution.

For each $i \in [k]$, the next follows  directly from Proposition~2 in \cite{lam2017optimization}.

\begin{lemma}
\label{eq:etaidifferentiable}
    Let \Cref{assump:input.length} 
    hold.
    Let $\bF$ be either $\bF^c$
    or $\widehat{\bF}$. 
    Then for each $i \in [k]$,
    the influence functions $\IF_{ip}(\cdot;\bF)$, 
    $p \in [m]$, are well-defined, have zero mean, and the identity in \eqref{eq:influence_function_derivative}
    holds true.
\end{lemma}

Next, we develop sufficient conditions to guarantee
the following convergence for the worst-case linearization errors
over the distributions contained in the ambiguity set $\mathcal{U}_\alpha$:
\begin{equation}
\label{eq:worst-case-error}
\sup\nolimits_{\bw \in \mathcal{U}_\alpha} |\eta_i(\bw)  - \eta^L_i(\bw; \bF^c)|
= O_p(n^{-1})
\quad \text{and} \quad 
\sup\nolimits_{\bw \in \mathcal{U}_\alpha} |\eta_i(\bw)  - \eta^L_i(\bw; \bFhat)|
= O_p(n^{-1}).
\end{equation}%
Namely, the linearization in \eqref{eq:nonparam.delta} is asymptotically accurate 
uniformly for all distributions in 
$\mathcal{U}_\alpha$.
Convergence rates for the worst-case
linearization errors are also established in
\cite{Owen1988,lam2017optimization,Duchi2021}.
The following statement is an adaptation of Proposition~3 in \cite{lam2017optimization}.

\begin{proposition}
\label{thm:worst-case-expansion-error}
Let $\alpha \in (0,1)$ and  $i \in [k]$. 
Suppose that \Cref{assump:input.length,assump:data_size} hold.
Then,
convergence for the worst-case 
errors in 
\eqref{eq:worst-case-error} hold true.
\end{proposition}

\section{CI Halfwidths Derivation}
\label{sec:optimization}

In this section, we use the ambiguity set $\mathcal{U}_\alpha$ defined in~\Cref{sec:as} for $q = k-1$ and the linear approximation $\eta_i^L$ of $\eta_i$ in~\Cref{sec:model} to derive the CI bounds that satisfy the asymptotic coverage hypothesis~\eqref{eq:asymptotic-mcb-ci-1}. This hypothesis requires $k-1$ simultaneous comparisons separately for each $i \in [k]$. 

Let $U_{i\ell,n}$ be the optimal value of the following program:
\begin{equation}\label{eq:U_iell.definition}
U_{i\ell,n} \triangleq \max\nolimits_{\bw\in\mathcal{U}_\alpha} \;\; \eta^L_i(\bw; \bF^c) - \eta^L_\ell(\bw; \bF^c).
\end{equation}
Namely, $U_{i\ell,n}$ is the worst-case (linearization-based) performance gap between Solutions $i$ and $\ell$, computed over the distributions $\bw$ contained in $\mathcal{U}_\alpha$. We also denote the maximizer of~\eqref{eq:U_iell.definition} by $\bw^*_{i\ell}$. 

Since $\mathcal{U}_\alpha$ contains probability measures 
with support $\{\bX_{pj}\}_{j \in [n_p]}$ and corresponding weights $\bw = \{w_{pj}\}_{j \in [n_p]}$, the integral 
in~\eqref{eq:nonparam.delta} becomes a weighted sum with weights $\bw\in\mathcal{U}_\alpha$:
\begin{small}
\[
    \eta^L_i(\bw; \bF^c) - \eta^L_\ell(\bw; \bF^c) = \eta_i(\bF^c) - \eta_\ell(\bF^c) + \sum\nolimits_{p=1}^m \sum\nolimits_{j=1}^{n_p} \big( \IF_{ip}(\bX_{pj}; \bF^c) - \IF_{\ell p}(\bX_{pj}; \bF^c)  \big)w_{pj}.
\] 
\end{small}%
To demonstrate that 
the collection, $\{U_{i\ell,n}\}_{i \neq \ell}$, satisfies the asymptotic coverage  hypothesis~\eqref{eq:asymptotic-mcb-ci-1}, let us first replace the generic vector, $\vect{Y}_{pj}$, introduced in~\Cref{sec:as} with
\begin{equation} \label{eq:redefined.Y}
\vect{Y}_{pj} = \left\{m^{-1}\big(\eta_i(\bF^c) - \eta_\ell(\bF^c)\big) + \IFXipj - \IFXlpj\right\}_{\ell\neq i}.
\end{equation} 
We observe that $\{\vect{Y}_{pj}\}_{p\in[m], j\in[n_p]}$ are mutually independent across $p$ and i.i.d.\ across $j$ for each $p$. Moreover, for each $p\in[m]$ and $j\in[n_p]$, $\vect Y_{pj}\in\mathbb{R}^{k-1}$ has mean $m^{-1}(\eta_i(\bF^c) - \eta_\ell(\bF^c))_{\ell \neq i}$ and the covariance matrix element corresponding to the two Solutions $\ell$ and $\ell^\prime$ is given by $\Cov(\IFXipj - \IFXlpj, \IFXipj - {\IF}_{\ell^\prime p}(\vect X_{pj}; \bF^c))$. The following assumption ensures that the resulting covariance matrix is positive definite. 
\begin{assumption}\label{assump:pos.def.cov}
    For each $p\in [m]$, the variance-covariance matrix of $\big(\IF_{ip}(\bX_{p};\bF^c)\big)_{i\in[k]}$ is positive definite.
\end{assumption}

\Cref{assump:pos.def.cov} essentially excludes the case where a linear combination of all $\IF_{ip}(\bX_{p};\bF^c), i\in[k]$ has a degenerate distribution, which is unlikely to occur in most  problems in practice.
Under \Cref{assump:data_size,assump:pos.def.cov}, $\{\vect{Y}_{pj}\}_{p\in[m], j\in[n_p]}$ satisfy the conditions of \Cref{thm:multisample-el,lem:mean.CR.coverage}.
Furthermore, $\bar{\bmu} = \eta_i(\bF^c) - \eta_\ell(\bF^c)$ and $\sum_{p=1}^m\sum_{j=1}^{n_p}\vect{Y}_{pj}w_{pj} = \eta^L_i(\bw; \bF^c) - \eta^L_\ell(\bw; \bF^c)$.
Consequently, we can redefine $C_\alpha$ in \eqref{eq:confidence-set} as 
\begin{equation}\label{eq:redefined.Ca}
    C_\alpha = \{(\eta_i^L(\vect{w}; \bF^c) - \eta_\ell^L(\vect{w}; \bF^c))_{\ell \neq i}: \vect{w}\in\mathcal{U}_\alpha\},
\end{equation}
and the optimal value $B_\ell$ in~\eqref{eq:generic_opt} corresponds to $U_{i\ell,n}$ in~\eqref{eq:U_iell.definition}. 
Therefore, we obtain the following corollary as a direct consequence of~\Cref{lem:mean.CR.coverage}.

\begin{corollary}\label{cor:upper_bound_opt}
Under \Cref{assump:input.length,assump:data_size,assump:pos.def.cov}, the collection $\{U_{i\ell,n}\}_{i\neq \ell}$ given in~\eqref{eq:U_iell.definition} satisfies 
\eqref{eq:asymptotic-mcb-ci-1}, that is,
$\liminf_{n\to\infty} \Pr\{ \eta_i(\bF^c) - \eta_\ell(\bF^c) \leq U_{i\ell,n}, \forall \ell \neq i \} \geq 1-\alpha$. 
\end{corollary}

Since $\IF_{ip}(\cdot;\bF^c)$ is generally
unknown, we cannot solve~\eqref{eq:U_iell.definition} to obtain $U_{i\ell,n}$. A natural estimator for $\IF_{ip}(\cdot;\bF^c)$ is $\IF_{ip}(\cdot;\bFhat)$, which replaces $\bF^c$ with the edf $\bFhat$.
Although $\IF_{ip}(\cdot;\bFhat)$ does not have an analytical form, it can be estimated via simulation. Under \Cref{assump:input.length}, Proposition~5 in \cite{lam2017optimization} 
shows that $\IFhat_{ip}(\bx, \bFhat)$ at $\bx=\bX_{pj}$ in the input data set can be rewritten as
\begin{equation} \label{eq:influence_function_cov}
\IFhatXipj = \Cov_{\bFhat}\Big(\, h_i(\matr{Z}^i_{1},\ldots,\matr{Z}^i_{m})\, ,\,  n_p  \sum\nolimits_{t=1}^{T_{ip}} \one{\bZ^i_{p}(t) = \bX_{pj}} - T_{ip} \,\Big).
\end{equation}

We have $\E_{\bFhat}[\one{\bZ^i_{p}(t) = \bX_{pj}}]=1/n_p$, which ensures the second term inside the covariance function in~\eqref{eq:influence_function_cov} has mean $0$. 
Since $\bFhat$ is available to us, we can estimate \eqref{eq:influence_function_cov} by running $R_1$ replications of Solution $i$. Let $\overline{h}_i \triangleq R_1^{-1} \sum\nolimits_{r=1}^{R_1} h_i(\matr{Z}^{ir}_{1},\ldots,\matr{Z}^{ir}_{m})$, where the superscript, $r$, in $\matr{Z}_{p}^{ir}$ indicates it is an input for the $r$th replication. 
We define our estimator of $\IFhat_{ip}(\cdot, \bFhat)$ by
\begin{equation} \label{eq:if_estimator}
\IFhathatXipj \triangleq R_1^{-1} \sum\nolimits_{r=1}^{R_1} \Big(h_i(\matr{Z}^{ir}_{1},\ldots,\matr{Z}^{ir}_{m}) - \overline{h}_i\Big)\Big(n_p \sum\nolimits_{t=1}^{T_{ip}} \one{\bZ_{p}^{ir}(t) = \bX_{pj}} - T_{ip}\Big).
\end{equation}
Because \eqref{eq:if_estimator} is a  sample covariance, it is a consistent estimator of  $\IFhat_{ip}(\bX_{pj}, \bFhat)$ given $\bFhat$ as $R_1 \to\infty$. We note that $\IFhathat(\bx;\bFhat)$ can be estimated simultaneously from $R_1$ replications. 
The conditional variance of $\IFhathatXipj$ given $\bFhat$ can be large when $\{T_{ip}\}_{p\in[m]}$ are large owing to the term, $\big(n_p \sum\nolimits_{t=1}^{T_{ip}} \one{\bZ_{p}^{ir}(t) = \bX_{pj}} - T_{ip}\big),$ in~\eqref{eq:if_estimator}, which necessitates choosing large $R_1$.  

Adopting \eqref{eq:if_estimator}, we define a new linear model
\[
\widehat{\eta}_i^L(\bw;\bFhat) \triangleq \eta_i(\bFhat) + \sum\nolimits_{p=1}^m \sum\nolimits_{j=1}^{n_p} \IFhathatXipj w_{pj}.
\]
There are two differences between $\eta_i^L(\bw;\bF^c)$ and $\widehat{\eta}_i^L(\bw;\bFhat)$; first, the constant term is $\eta_i(\bFhat)$, not $\eta_i(\bF^c)$ as the linear approximation is made at $\bFhat$; second, it replaces $\IFXipj$ in $\eta_i^L(\bw; \widehat{\bF})$ with $\IFhathatXipj$. 
Below, we discuss how $\widehat{\eta}_i^L(\bx;\bFhat)$ can be utilized to approximate $\{U_{i\ell,n}\}_{i\neq \ell}$. 

Let us define $\widehat\bw_{i\ell}^*$ as the maximizer to a modified version of~\eqref{eq:U_iell.definition} that replaces the objective function, $\eta_i^L(\bw;\bF^c)-\eta_i^L(\bw;\bF^c)$, with $\widehat{\eta}_i^L(\bw;\bFhat)-\widehat{\eta}_\ell^L(\bw;\bFhat)$. Because adding a constant to the objective function does not affect the optimality of $\widehat\bw_{i\ell}^*$, we can rewrite the modified program as
\begin{equation}\label{eq:opt_dro3}
    \max\nolimits_{\vect{w} \in \mathcal{U}_\alpha} \quad  \sum\nolimits_{p=1}^m \sum\nolimits_{j=1}^{n_p} \big(\IFhathatXipj - \IFhathatXlpj  \big) w_{pj},
\end{equation}
dropping the constant, $\eta_i(\bFhat)-\eta_\ell(\bFhat)$, from its objective. 
Program~\eqref{eq:opt_dro3} is a convex  problem; $\mathcal{U}_\alpha$ is a convex set defined by explicit convex constraints and its objective is linear in $\vect{w}$. 
Thus,~\eqref{eq:opt_dro3} can be solved efficiently by a convex optimization algorithm within a user-specified suboptimality.

Even if $\widehat{\bw}^{*}_{i\ell}$ is obtained by solving~\eqref{eq:opt_dro3}, we cannot estimate $U_{i\ell, n}$ by computing $\widehat{\eta}_i^L(\widehat{\bw}^{*}_{i\ell})$ 
since the difference $\eta_i(\bFhat) - \eta_\ell(\bFhat)$ is unknown. Instead, we compute the following estimator of ${U}_{i\ell,n}$ by simulating $R_2$ replications of Solutions $i$ and $\ell$ with $\widehat{\bw}^{*}_{i\ell}$ as the input model:
\begin{equation}\label{eq:U_iell.estimate}
    \widehat{U}_{i\ell,n} \triangleq {R_2}^{-1} \sum\nolimits_{r=1}^{R_2} \big(Y_{ir}(\widehat{\bw}^{*}_{i\ell}) - Y_{\ell r}(\widehat{\bw}^{*}_{i\ell})\big), 
\end{equation}
where $Y_{ir}(\widehat{\bw}^{*}_{i\ell})$ is the $r$th simulation output of Solution $i$ when the inputs are generated by $\widehat{\bw}^{*}_{i\ell}$. 

To summarize, we approximate $U_{i\ell,n}$ with $\widehat U_{i\ell,n}$ by estimating $\IFXipj$ with $\IFhathatXipj$, solving~\eqref{eq:opt_dro3} to find $\widehat{\bw}_{i\ell}^*$, and computing the empirical estimate of $\eta_i(\widehat{\bw}^{*}_{i\ell}) - \eta_\ell(\widehat{\bw}^{*}_{i\ell})$ in~\eqref{eq:U_iell.estimate}. These steps involve four errors: (i) finite-sample error of using the linear expansion at $\bFhat$ instead of $\bF^c$; (ii) estimation error in the influence function from finite $R_1$; and (iii) Monte Carlo error from finite $R_2$. Lastly, although subtle, (iv) $\widehat{U}_{i\ell,n}$ is an unbiased estimator of $\eta_i(\widehat{\bw}^{*}_{i\ell}) - \eta_\ell(\widehat{\bw}^{*}_{i\ell})$ given $\widehat{\bw}^*_{i\ell}$, not  $\eta_i^L(\widehat{\bw}^{*}_{i\ell}) - \eta_\ell^L(\widehat{\bw}^{*}_{i\ell})$, which can be linked to the worst-case linearization error studied in \Cref{sec:model}. 
To bound these errors, we first need to make some assumptions on the simulation output function.
The expectation with respect to  \(\bw \in \mathcal{U}_\alpha\) is denoted by \(\E_{\bw}[\cdot]\).

\begin{assumption}\label{assump:sim_output_bounded}
    For each $i \in [k]$, 
    $\sup_{\bw \in \mathcal{U}_\alpha} \E_{\bw}\big[ |h_i(\matr{Z}^{i}_{1},\ldots,\matr{Z}^{i}_{m})|^4 \big] < \infty$.
    In addition, for each $i \in [k]$
    and every sequence $I^p = (I^p_1,\ldots,I^p_{T_{ip}})$ 
    with $1 \leq I^p_t \leq T_{ip}$, and $\bbZ^i_{p,I^p} = \big(\bZ^i_{p}(I^p_1),\ldots,\bZ^i_{p}(I^p_{T_{ip}})\big)$, 
    we have $\E_{\bF^c}[|h_i(\bbZ^i_{1,I^1},\ldots,\bbZ^i_{m,I^m})|^4] < \infty$.
\end{assumption}

\Cref{assump:sim_output_bounded} implies that the fourth-order  moment of all simulation outputs are uniformly bounded over all $\bw\in \mathcal{U}_\alpha$. This boundedness holds whether input samples are reused, independently drawn, or a mix of both. 
It is essential for controlling simulation error in influence function estimation, where empirical distributions may yield repeated input combinations. If each $h_i$ is bounded, this condition is automatically satisfied.

\Cref{prop:algorithm-error} characterizes the asymptotic convergence rate of the approximation error in $\widehat U_{i\ell,n}$.

\begin{proposition}\label{prop:algorithm-error}
    Under \Cref{assump:data_size,assump:input.length,assump:pos.def.cov,assump:sim_output_bounded}, for $i \neq \ell$, 
    $
        |\widehat{U}_{i\ell,n} - U_{i\ell,n}| = o_p(n^{-1/2}) + O_p(R_1^{-1/2}) + O_p(R_2^{-1/2}).
    $
\end{proposition}

Let us revisit the discussion on the error sources (i)--(iv) in $U_{i\ell,n}$. The proof of \Cref{prop:algorithm-error} establishes that errors (i) and (iv) decrease at a rate faster than $n^{-1/2}$, while  (ii) and (iii)  diminish at rates $R_1^{-1/2}$  and $R_2^{-1/2}$, respectively. In our setting, the input data is more scarce than simulation data. Hence, it is reasonable to choose $R_1$ and $R_2$ large enough to avoid slowing the convergence rate of $\widehat{U}_{i\ell,n}$.

We close this section by stating that, despite the estimation error in $\{\widehat{U}_{i\ell,n}\}_{i\neq \ell}$, they can provide asymptotically correct CI bounds satisfying~\eqref{eq:asymptotic-mcb-ci-1} if $R_1$ and $R_2$ increase faster than $n$.
\begin{proposition}\label{prop:Uhat-coverage} 
Suppose that
\Cref{assump:input.length,assump:data_size,assump:pos.def.cov,assump:sim_output_bounded} hold. If $R_1$ and $R_2$ are chosen such  that $R_1/n\to\infty$ and $R_2/n \to\infty$, then $\liminf_{n\to\infty} \Pr\{ \eta_i(\bF^c) - \eta_\ell(\bF^c) \leq \widehat{U}_{i\ell,n}, \forall \ell \neq i \} \geq 1-\alpha$. 
\end{proposition}

\section{NIOU-C  Algorithm}\label{sec:as-algorithm}

\Cref{alg:as} outlines the NIOU-C procedure that computes the confidence set, $\widehat{I}_n$, by combining the simulation estimation and solving convex optimization problems as discussed in \Cref{sec:optimization}.

{
\centering
\begin{minipage}{0.9\textwidth}
\begin{algorithm}[H]
\small 
\caption{NIOU-C Algorithm}\label{alg:as}
\begin{algorithmic}[1]
    \State From i.i.d.\ observations $\bX_{p1},\ldots,\bX_{pn_p}$, define the edf, $\Fhat_p$, as the estimator of $F_p^c$ for each $p \in [m]$.
    \For{each $i \in [k]$}  \label{step:first.for.begins}
        \State Run $R_1$ replications of the simulation to generate $Y_{i1}(\bFhat),\ldots,Y_{iR_1}(\bFhat)$. \label{alg-step:r1} 
        \State Compute $\IFhathatXipj$ in \eqref{eq:if_estimator} for each $p \in [m]$ and $j \in [n_p]$.
    \EndFor \label{step:first.for.ends}
    \For{each $i \in [k]$} \label{step:second.for.begins}
        \For{each $\ell \neq i$}
        \State Solve \eqref{eq:opt_dro3} to find its optimum $\widehat{\bw}^{*}_{i\ell}$. \label{step:solve.DRO}
        \State Run $R_2$ simulations of Solutions $i$ and $\ell$ to obtain $\{Y_{ir}(\widehat{\bw}^{*}_{i\ell})\}_{r \in R_2}$ and $\{Y_{\ell r}(\widehat{\bw}^{*}_{i\ell})\}_{r \in R_2}$. \label{alg-step:r2} 
        \State Compute $\widehat{U}_{i\ell,n}$ in \eqref{eq:U_iell.estimate}. 
        \EndFor \label{step:second.for.ends}
    \EndFor
    \For{each $i \in [k]$}
        \State Compute $\widehat{D}_{i,n}^+ = \big(\min_{\ell \neq i} \{\widehat{U}_{i\ell,n}\}\big)^+$ and $\widehat{I}_n = \{i \in [k] \colon \widehat{D}_{i,n}^+ > 0\}$. Let 
        $\widehat{D}_{i,n}^- = 0$, if  $\widehat{I}_n = \{i\}$, and $\widehat{D}_{i,n}^- = -\big(\min_{\ell \in \widehat{I}_n \colon \ell \neq i} \{-\widehat{U}_{\ell i,n}\}\big)^-$, otherwise. \label{step:MCB.CIs}
    \EndFor
    \State \Return $[\widehat{D}_{i,n}^-, \widehat{D}_{i,n}^+]$ for each $i \in [k]$ and $\widehat{I}_n$.
\end{algorithmic}
\end{algorithm}
\end{minipage}
}
\\

NIOU-C first runs $R_1$ replications at all $k$ solutions to estimate their influence functions. In our experiments, we use common random numbers (CRNs) across all solutions when generating the inputs in~\eqref{eq:if_estimator} for all our experiments in \Cref{sec:empirical}. That is, for each $r \in [R_1]$ and $p \in [m]$, $\bZ_{p}^{ir}(1),\ldots,\bZ_{p}^{ir}(T_{ip})$ are reused across all solutions as much as possible. If $T_{ip}$ remains the same for all solutions, then all $T_{ip}$ inputs can be reused. 
Provided that the CRNs induce positive correlations among the solutions' simulation outputs, adopting them may reduce the variance in our framework. We recall that the influence function estimators always enter our framework as the pairwise difference between two solutions. In the simplest case when $T_{ip} = T_{\ell p}$ for all $i\neq \ell$, 
it follows from~\eqref{eq:if_estimator} that
$\IFhathatXipj - \IFhathat_{ip}(\vect{X}_{pj};\bFhat)$ is equal to
\[
R_1^{-1} \sum\nolimits_{r=1}^{R_1} \Big(h_i(\matr{Z}^{ir}_{1},\ldots,\matr{Z}^{ir}_{m}) - h_\ell(\matr{Z}^{\ell r}_{1},\ldots,\matr{Z}^{\ell r}_{m}) - \overline{h}_i + \overline{h}_\ell\Big)\Big(n_p \sum\nolimits_{t=1}^{T_{ip}} \one{\bZ_{p}^{ir}(t) = \bX_{pj}} - T_{ip}\Big).
\]
In contrast to~\eqref{eq:if_estimator}, the first term in each summand contributes to variance reduction.

Because we solve~\eqref{eq:opt_dro3} for all $\ell \neq i$ at each $i\in [k]$, the total simulation budget for estimating $\{\widehat{U}_{i\ell,n}\}_{i\neq \ell}$ in Step~\ref{alg-step:r2}  is $2k(k-1)R_2$. Note that symmetry does not hold here, i.e.,  $\bw_{i\ell}^* \neq \bw_{\ell i}^*$ in general.
Since $\widehat{U}_{i\ell,n}$ is computed as the difference of the sample average estimators for Solutions $i$ and $\ell$, we also adopt the CRNs to reduce the variability in the estimation.
In summary, the total simulation cost of \Cref{alg:as} is $kR_1 + 2k(k-1)R_2$. 

NIOU-C is conducive to parallelization. The first for loop in Steps~\ref{step:first.for.begins}--\ref{step:first.for.ends} can be parallelized across $k$ solutions 
by synchronizing the seed in each replication even if the CRNs are adopted. 
The double for loop in Steps~\ref{step:second.for.begins}--\ref{step:second.for.ends} can be parallelized across all $k(k-1)$ ordered pairs of solutions. 

Lastly, we state 
\Cref{thm:as-asymp-coverage} ensures that NIOU-C provides the correct asymptotic probability guarantee when the numbers of simulations, $R_1$ and $R_2$, grow  faster than $n$:
\begin{theorem}
\label{thm:as-asymp-coverage}
    Under \Cref{assump:data_size,assump:pos.def.cov,assump:input.length,assump:sim_output_bounded}, 
    for $\widehat{D}_{i,n}^-$, $\widehat{D}_{i,n}^+$ and $\widehat{I}_n$ computed by \Cref{alg:as},
    \begin{equation*}
        \liminf\nolimits_{n\to\infty, \, R_1/n\to\infty, \, R_2/n\to\infty}
        \Pr\{\eta_i(\bF^c) - \max\nolimits_{\ell \neq i}\eta_\ell(\bF^c) \in [\widehat{D}_{i,n}^-, \widehat{D}_{i,n}^+], \forall i \in [k]\} \geq 1 - \alpha,
    \end{equation*}
    and asymptotic guarantee in \eqref{eq:asymptotic-coverage} holds true for $I_n = \widehat{I}_n$.
\end{theorem}

\section{Extension with Tighter Bounds}\label{sec:extension}

We computed a polyhedral set containing $C_\alpha$ (see \eqref{eq:redefined.Ca}) using the optimal values $\{U_{i\ell,n}\}_{i\neq \ell}$ from \eqref{eq:U_iell.definition}. While this method ensures coverage, it becomes increasingly conservative as  $k$ grows, leading to overcoverage of the joint CI in~\eqref{eq:asymptotic-mcb-ci-1}. 
To mitigate this, we propose an alternative approach for computing  $\{U_{i\ell,n}\}_{i\neq \ell}$. We adapt the technique developed in Section~5.1 of \cite{lam2019}, originally developed for tight statistical guarantees in EL-based DRO with multiple expectation constraints.

As shown in \Cref{thm:multisample-el}, twice the negative log empirical likelihood under the true mean constraint converges to $\chi^2_q$. To reduce conservatism, \cite{lam2019} models the $q$ dimensions as correlated $\chi^2_1$ variables and uses the $1-\alpha$ quantile of their maximum---rather than $\chi^2_{q,1-\alpha}$---to construct a tighter confidence region for the mean vector in \eqref{eq:uncertaintyset}.
To apply this modification to NIOU-C, we estimate the corresponding quantile for each Solution $i$. 
Let us define 
$
\upsilon_{\ell}(\bX_{pj}) \triangleq \IFXipj - \IFXlpj 
$
and let $q_n$ be the $1-\alpha$ quantile of $\sup_{\ell \neq i} J_\ell$, where for each $\ell \neq i$, $J_\ell \triangleq G_\ell^2$, $G_\ell\sim \N(0,1)$, and
\begin{equation} \label{eq:multivariate.norm.cov}
    \mathrm{Corr}(G_{\ell}, G_{\ell^\prime}) = \frac{\sum_{p=1}^m \frac{n_p}{n_p-1}\Cov(\upsilon_{\ell}(\bX_p), \upsilon_{\ell^\prime}(\bX_p))}{\sqrt{\sum_{p=1}^m \frac{n_p}{n_p-1}\V(\upsilon_{\ell}(\bX_p))\sum_{p=1}^m \frac{n_p}{n_p-1}\V(\upsilon_{\ell^\prime}(\bX_p))}}, \quad \ell^\prime \in [k]\backslash \{i,\ell\}.
\end{equation}
Marginally, $J_\ell \sim \chi^2_1$ and \eqref{eq:multivariate.norm.cov} describes the correlations among the $k-1$ such $\chi^2_1$ random variables.  The variance and covariance terms in~\eqref{eq:multivariate.norm.cov} must be estimated from our influence function estimators. 

\Cref{alg:quantile} outlines the Monte Carlo estimation of $q_n$. We replace~\eqref{eq:multivariate.norm.cov} with its estimate $\widehat{R}_{\ell,\ell^\prime}$. 
Since the influence function estimator from the main experiment is utilized here, its estimation error induces bias into $\widehat{R}_{\ell,\ell^\prime}$. This may affect the finite-sample coverage of NIOU-C, when $\widehat{q}_n$ replaces $\chi^2_{k-1,1-\alpha}$. We name this version of the modified algorithm as the NIOU-C:E, where E stands for extension and compare its empirical performances against NIOU-C in \Cref{sec:empirical}.

{
\centering
\begin{minipage}{0.9\textwidth}
\begin{algorithm}[H]
\small 
\caption{Quantile estimation for extension of NIOU-C}\label{alg:quantile}
\begin{algorithmic}[1]
    \State Let $\widehat{\upsilon}_{\ell}(\bX_{pj}) \triangleq \IFhathatXipj - \IFhathatXlpj$.
    For each input $p\in [m]$ and each $\ell\in[k]\setminus i$, compute 
    \[
    \widehat{\mu}_{\ell}^{(p)} = \frac{1}{n_p}\sum_{j=1}^{n_p} \widehat{\upsilon}_{\ell}(\bX_{pj}), \qquad
    \hat{s}_{\ell\ell}^{(p)} = \frac{1}{n_p-1}\sum_{j=1}^{n_p} \big(\widehat{\upsilon}_{\ell}(\bX_{pj}) - \widehat{\mu}_{\ell}^{(p)}\big)^2.
    \]
    For each pair $(\ell,\ell')$, compute the covariance for each $p$,
    $
    \hat s_{\ell\ell'}^{(p)} = (n_p-1)^{-1}\sum_{j=1}^{n_p} 
    \big(\widehat{\upsilon}_{\ell}(\bX_{pj}) - \widehat{\mu}_{\ell}^{(p)}\big)
    \big(\widehat{\upsilon}_{\ell'}(\bX_{pj}) - \widehat{\mu}_{\ell'}^{(p)}\big).
    $
    
    \State 
    Let $N=\sum_{p=1}^m n_p$. Define the covariance estimator
    $
    \widehat{\Sigma}_{\ell\ell'} = N^{-1}\sum_{p=1}^m n_p\, \hat s_{\ell\ell'}^{(p)}.
    $
    
    \State 
    Estimate~\eqref{eq:multivariate.norm.cov} with 
    $\widehat{R} = (\widehat{R}_{\ell\ell'})_{1\le \ell,\ell'\le d}$, where
    \[
    \widehat{R}_{\ell\ell'} = \frac{\widehat\Sigma_{\ell\ell'}}{\sqrt{\widehat\Sigma_{\ell\ell}\,\widehat\Sigma_{\ell'\ell'}}}
    =
    \frac{ \sum_{p=1}^m \frac{n_p}{n_p - 1} \sum_{j=1}^n \big(\widehat{\upsilon}_\ell(\bX_{pj}) - \widehat{\mu}^{(p)}_\ell\big)\big(\widehat{\upsilon}_{\ell'}(\bX_{pj}) - \widehat{\mu}^{(p)}_{\ell'}\big)}{\sqrt{\sum_{p=1}^m \frac{n_p}{n_p - 1}\sum_{j=1}^n \big(\widehat{\upsilon}_\ell(\bX_{pj}) - \widehat{\mu}^{(p)}_\ell\big)^2  \sum_{p=1}^m \frac{n_p}{n_p - 1}\sum_{j=1}^n \big(\widehat{\upsilon}_{\ell'}(\bX_{pj}) - \widehat{\mu}^{(p)}_{\ell'}\big)^2}}.
    \]

    \State \label{alg:ext-step4}
    Let $\vect{\xi} \sim \mathrm{N}(0,\widehat{R})$. 
    Generate $M$ independent realizations $\vect{\xi}^{(1)},\dots,\vect{\xi}^{(M)}$, each of dimension $k-1$.
    
    \State 
    For each $t\in[M]$, compute
    $
        S_t = \max_{\ell \neq i} \big(\xi^{(t)}_\ell\big)^2.
    $
    
    \State 
    Let $\widehat{q}_n$ be the empirical $(1-\alpha)$ quantile of $\{S_t : t\in[M]\}$.
    
    \State \Return $\widehat{q}_n$.
\end{algorithmic}
\end{algorithm}
\end{minipage}
}\\

\section{Empirical Study}
\label{sec:empirical}
We empirically evaluate the performances of NIOU-C and NIOU-C:E using various examples. 
\Cref{subsec:empirical-sens-analysis} presents a numerical example where $h_i$ is chosen so that both $\eta_i$ and its influence functions with respect to all input distributions can be derived analytically.
This example examines how approximation errors impact the NIOU-C framework under varying input sizes, simulation budgets, and the number of solutions.  
\Cref{subsec:empirical-queue-system} considers a queueing system capacity control problem using a discrete-event simulator, benchmarking NIOU-C against IOU-C from \cite{Song2019} that requires assuming the parametric distributions for the input models.

To compute $\widehat{q}_n$ required for NIOU-C:E, we set $M=20\text{,}000$ in Step~\ref{alg:ext-step4} of  \Cref{alg:quantile}.
All experiments are repeated for 1000 macro-runs to estimate:
(i) MCB Coverage probability in~\eqref{eq:asymptotic-mcb-ci-2}; (ii)  $\Pr\{i^c \in \widehat{I}_n\}$; and (iii)  average size of $\widehat{I}_n$, $\widehat\E|\widehat{I}_n|$. For the examples in \Cref{subsec:empirical-sens-analysis}, we also report (iv) the average MCB CI width for $i^c$, $\widehat\E[\widehat{D}^+_{i^c,n}-\widehat{D}^-_{i^c,n}]$.
The experiments were run on the Phoenix cluster \citep{pace} with Intel Xeon Gold 6226 CPUs (2.7 GHz, 24 cores) and
192 GB of RAM\@.
We use \texttt{CVXPY} \citep{diamond2016cvxpy} and \texttt{SCS} \citep{odonoghue:21} to solve \eqref{eq:opt_dro3} with optimality tolerance $10^{-6}$.

\subsection{Numerical examples with known response functions}\label{subsec:empirical-sens-analysis}

For each $p \in [m]$, let $F_p^c=\N(c_p, \tau^2_p)$. For all $p \neq p^\prime$, $Z_p$ is independent of $Z_{p^\prime}$. Let $\vect{c} \triangleq (c_1,\ldots,c_m)$, $\vect{\tau}^2 \triangleq (\tau^2_1,\ldots,\tau^2_m)$ and $\vect{a}\triangleq (a_1,\ldots,a_k)$, where $a_i\in\mathbb{N}$ characterizes Solution $i$. 
For each $i \in [k]$, 
the simulation output is given by
\begin{equation}\label{eq:expr.output}
    h_i(\vect{Z}) = (1/T) \sum\nolimits_{t=1}^T \left[ \sum\nolimits_{p=1}^m a_i (Z_p{(t)} - a_i)^2 + \sum\nolimits_{p=1}^m \sum\nolimits_{{p^\prime}>p} a_i (Z_p{(t)} - a_i) (Z_{p^\prime}{(t)} - a_i) \right],
\end{equation}
where $(Z_p{(1)}, \ldots, Z_p{(T)})$ are the random variates drawn from $F_p$, and  $T_{ip}=T$ for all $i \in [k]$ and $p \in [m]$. Therefore,
\begin{equation*}
    \eta_i(\bF) = \E\Big[ \sum\nolimits_{p=1}^m a_i (Z_p - a_i)^2 + \sum\nolimits_{p=1}^m \sum\nolimits_{{p^\prime}>p} a_i (Z_p - i) (Z_{p^\prime} - a_i) \Big]
\end{equation*}
and we can analytically derive
\[
\IF_{ip}(x; \bF^c) = a_i (x - a_i)^2 - a_i \tau^2_p - a_i (c_p - a_i)^2 + \sum\nolimits_{{p^\prime} \neq p} a_i (x - c_p) (c_{p^\prime} - a_i).
\]
In the following, we consider different  configurations with varying degrees of difficulty in identifying $i^c$.
In all experiments,  $T=10$, $m=2$, $\vect{c}=(193, 200)$, $\vect{\tau}^2=(1833, 2000)$,
and $n_1 = n_2 = n$. The coverage probability is  $1-\alpha=0.9$,  and CRNs are used for all solutions as discussed in \Cref{sec:as-algorithm}.

\subsubsection{Case 1}

\begin{figure}[tbp]
    \centering
    \includegraphics[width=0.9\textwidth]{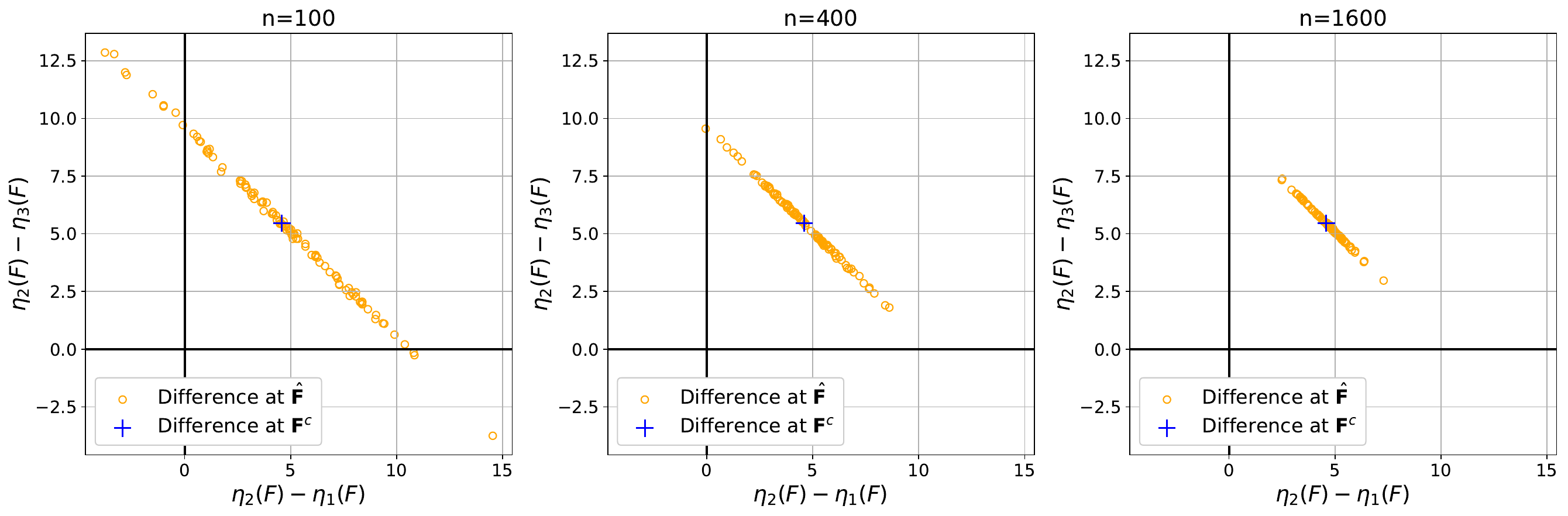}
\caption{
Comparison of the mean differences for $\bF^c$ and $\bFhat$ in Case 1 for different input data size $n$.
\label{fig:case1-difficulty}
The cross represents the true mean differences and the dots represents different realizations of $\bFhat$.
}
\end{figure}

We consider $k=3$ solutions with $\vect{a}=(66, 69, 72)$, yielding $i^c = 2$.
\Cref{fig:case1-difficulty} shows how inferring $i^c$ becomes easier as the sample size $n$ grows from $100$ to $400$ and $1600$. The true mean differences $(\eta_2(\bF^c) - \eta_1(\bF^c), \eta_2(\bF^c) - \eta_3(\bF^c))$ are marked with a cross. For each $n$, we generate $100$ independent samples from $\bF^c$ and plot $(\eta_2(\widehat{\bF}) - \eta_1(\widehat{\bF}), \eta_2(\widehat{\bF}) - \eta_3(\widehat{\bF}))$. A point in the first quadrant indicates $i^c = 2$ remains the conditional best. As $n$ increases, points concentrate in the first quadrant: 88\% for $n=100$, 99\% for $n=400$, and 100\% for $n=1600$.


\Cref{fig:if-estimation} highlights benefits of using CRNs  over independent runs for influence function estimation.  For five macro-runs, the solid line depicts $\IFhat_{2,1}(\cdot;\bFhat) - \IFhat_{3,1}(\cdot;\bFhat)$ and the dotted lines depict $\IFhathat_{2,1}(\cdot;\bFhat) - \IFhathat_{3,1}(\cdot;\bFhat)$ evaluated at data  $\{X_{1j}\}_{j \in [100]}$ from the first input model. The estimation error is orders of magnitude larger when CRNs are not used. Although this effect may be more pronounced for the output function in~\eqref{eq:expr.output}, we observed similar phenomena for others.

\begin{figure}[tbp]
    \centering
    \includegraphics[width=.8\textwidth]{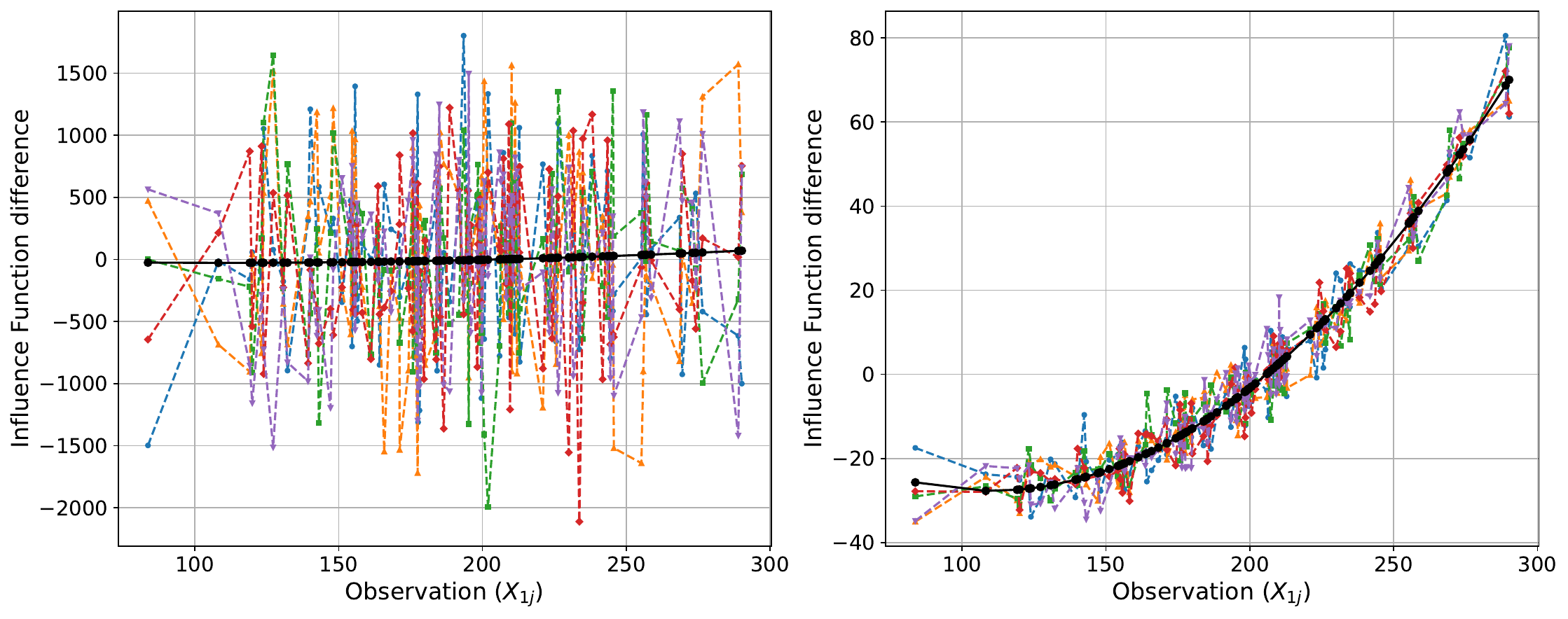}
\caption{
The influence function estimates using independent simulations (left) and  CRNs are used (right).
\label{fig:if-estimation}
The solid line is $\IFhat_{2,1}(\cdot;\bFhat) - \IFhat_{3,1}(\cdot;\bFhat)$; each dotted line is $\IFhathat_{2,1}(\cdot;\bFhat) - \IFhathat_{3,1}(\cdot;\bFhat)$ from a different macro-run.
}
\end{figure}

\Cref{tab:case1} compares  NIOU-C and NIOU-C:E for $n=100$, $400$, and $1600$, each using four choices  of $(R_1, R_2)$ to examine the effect of increasing the simulation budget. The first and last settings allocate equal simulation budget for estimating the influence functions and computing  $\widehat{U}_{i\ell,n}$. 
We also define a \textit{benchmark} value for each performance metric by applying the procedure to the same problem, but using $\eta_i(\cdot)$ and $\IF_{ip}(\cdot;\bF^c)$ instead of their estimates. 
The benchmark assesses the impact of  finite budget, but still reflects  IU.

\begin{table}[tbp]
    \centering
    \caption{Performance comparison of NIOU-C and NIOU-C:E for Case 1 from 1000 macro-runs.\label{tab:case1}}
    {\resizebox{0.9\textwidth}{!}{%
        \begin{tabular}{|c|c||cccc||cccc|}
            \hline
            & & \multicolumn{4}{c||}{NIOU-C} & \multicolumn{4}{c|}{NIOU-C:E}\\\hline
            $n$ & $(R_1,R_2)$ & MCB Cov. & $\Pr\{i^c \in \widehat{I}_n\}$  & $\widehat\E|\widehat{I}_n|$ &  $\widehat\E[\widehat{D}^+_{i^c,n}-\widehat{D}^-_{i^c,n}]$ & MCB Cov. & $\Pr\{i^c \in \widehat{I}_n\}$ & $\widehat\E|\widehat{I}_n|$ & $\widehat\E[\widehat{D}^+_{i^c,n} -\widehat{D}^-_{i^c,n}]$
            \\\hline\hline
             \multirow{5}{*}{100}
             & (100, 25) & 0.689 & 0.976 & 1.833 & 8.7 & 0.570 & 0.966 & 1.649 & 6.9 \\
             & (100, 100) & 0.746 & 0.991 & 1.824 & 8.7 & 0.631 & 0.976 & 1.601 & 7.0 \\
             & (400, 25) & 0.861 & 0.993 & 2.179 & 12.2 & 0.735 & 0.986 & 1.914 & 9.4 \\
             & (400, 100) & 0.913 & 0.996 & 2.236 & 12.1 & 0.794 & 0.994 & 1.897 & 9.4 \\\cline{2-10}
             & Benchmark & 0.973 & 0.999 & 2.514 & 14.9 & 0.910 & 0.998 & 2.169 & 11.4 \\\hline\hline
             \multirow{5}{*}{400}
             & (400, 100) & 0.716 & 0.997 & 1.159 & 5.7 & 0.593 & 0.997 & 1.100 & 5.2 \\
             & (400, 400) & 0.775 & 1.000 & 1.114 & 5.8 & 0.650 & 0.998 & 1.066 & 5.3 \\
             & (1600, 100) & 0.861 & 0.999 & 1.342 & 6.8 & 0.757 & 0.997 & 1.186 & 5.9 \\
             & (1600, 400) & 0.911 & 1.000 & 1.281 & 6.9 & 0.814 & 1.000 & 1.140 & 6.0 \\\cline{2-10}
             & Benchmark & 0.966 & 1.000 & 1.470 & 7.8 & 0.910 & 1.000 & 1.207 & 6.7  \\\hline\hline
             \multirow{5}{*}{1600}
             & (1600, 400) & 0.736 & 1.000 & 1.000 & 5.2 & 0.600 & 1.000 & 1.000 & 5.0 \\
             & (1600, 1600) & 0.776 & 1.000 & 1.000 & 5.3 & 0.665 & 1.000 & 1.000 & 5.1 \\
             & (6400, 400) & 0.876 & 1.000 & 1.001 & 5.7 & 0.786 & 1.000 & 1.000 & 5.3 \\
             & (6400, 1600) & 0.922 & 1.000 & 1.000 & 5.8 & 0.811 & 1.000 & 1.000 & 5.4 \\\cline{2-10}
             & Benchmark & 0.971 & 1.000 & 1.001 & 6.2 & 0.923 & 1.000 & 1.001 & 5.7  \\\hline
        \end{tabular}
    }}
    {
    }
\end{table}

For $n=100$,  the benchmarks of both NIOU-C and NIOU-C:E meet the target MCB coverage. NIOU-C shows a more conservative behavior, as evidenced by higher MCB coverage and $\Pr\{i^c \in \widehat{I}_n\}$, resulting in a larger $\widehat{I}_n$ and a wider MCB CI  compared to NIOU-C:E\@.
As $n$ increases to $400$ and $1600$, both benchmarks return increasingly  smaller $\widehat{I}_n$ and include $i^c$ in $\widehat{I}_n$ in all 1000 macro-runs as expected from \Cref{fig:case1-difficulty}.

With simulation error, both NIOU-C and NIOU-C:E include $i^c$ in $\widehat{I}_n$ with or above probability 0.9 in all $(R_1,R_2)$ settings. Meanwhile, the MCB coverage falls  below 0.9 for smaller $R_1$ and $R_2$, particularly for NIOU-C:E\@. The undercoverage of NIOU-C:E can be attributed to that it is design to cut conservatism by estimating the pairwise correlation of the $\chi^2_1$ random variables in~\eqref{eq:multivariate.norm.cov}. Once the simulation error is introduced, the exact $\text{IF}_{ip}$ is  replaced by its estimator, increasing the estimation error of the correlation matrix that can lead to imprecise estimation of $\widehat{q}_n$ in \Cref{alg:quantile} and  undercoverage. Moreover, for both NIOU-C and NIOU-C:E, recall that $\widehat{U}_{i\ell,n}$ for all $i\neq \ell$ are estimated from $R_2$ simulations, which may cause undercoverage in both for small $R_2$, although its effect on NIOU-C may be masked by its conservatism. Indeed,  \Cref{tab:case1} shows that increasing $R_2$ improves the MCB coverage and $\Pr\{i^c \in \widehat{I}_n\}$ for both algorithms.

When $R_1$ increases, influence function estimators become more precise and
all four metrics increase for both algorithms. 
The CI width, $\widehat{D}^+_{i^c,n}-\widehat{D}^-_{i^c,n}$, increasing in $R_1$  may appear counterintuitive, as one typically expects a CI to shrink with increasing simulation effort. We can link this observation to the so-called optimizer's curse~\citep{smith2006}. 
The maximizer, $\widehat{\bw}^{*}_{i\ell}$, of \eqref{eq:opt_dro3}, is suboptimal for~\eqref{eq:U_iell.definition}  owing to the error in approximating $\eta_i^L(\vect{w};\bFhat)$ with $\widehat{\eta}_i^L(\vect{w};\bFhat)$. This implies that the objective  of~\eqref{eq:U_iell.definition} at $\widehat{\bw}^{*}_{i\ell}$ is smaller than its optimal value, $U_{i\ell,n}$, ultimately making the CI width, $\widehat{D}_{i^c,n}^+ - \widehat{D}_{i^c,n}^-$, smaller. When $R_1$ is larger, $\widehat{\eta}_i^L(\vect{w};\bFhat)$ better approximates $\eta_i^L(\vect{w};\bFhat)$. Hence $\widehat{\bw}^{*}_{i\ell}$ converges to a maximizer of \eqref{eq:opt_dro3} and widens $[\widehat{D}_{i^c,n}^-, \widehat{D}_{i^c,n}^+]$. At the same time, $\widehat{I}_n$ becomes larger and the coverage improves.

\subsubsection{Case 2}

\begin{figure}[t]
\centering
    \includegraphics[width=0.9\textwidth]{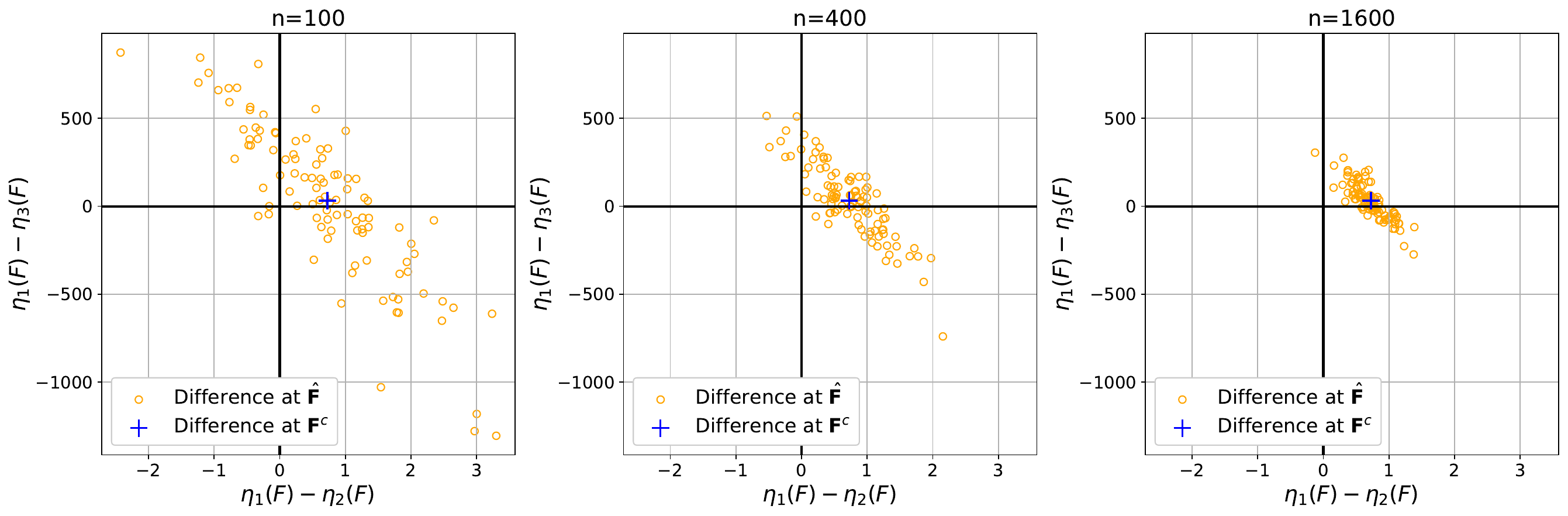}
\caption{
Comparison of the mean differences for $\bF^c$ and $\bFhat$ in Case 2 for different input data size.
\label{fig:case2-difficulty}
The cross represents the true mean differences and the dots represents different realizations of $\bFhat$.
}
\end{figure}

We consider a more difficult case with $k=3$  and $\vect{a}=(69, 70, 255)$, which yields $i^c = 1$. 
\Cref{fig:case2-difficulty} illustrates the difficulty of inferring $i^c$ as $n$ increases.
Compared to Case~1, the nonlinearity of the pairwise mean differences at each  $\bFhat$ is
more pronounced here. When $n=100$, 32\% of the points are in the first quadrant, which increases to 50\% and 57\% for $n=400$ and $n=1600$, respectively,  suggesting that identifying $i^c$ is challenging even for large $n$.

\begin{table}[t]
    \centering
    \caption{Performance comparison of NIOU-C and NIOU-C:E for Case 2 from 1000 macro-runs.\label{tab:case2}}
    {\resizebox{0.9\textwidth}{!}{%
        \begin{tabular}{|c|c||cccc||cccc|}
            \hline
            & & \multicolumn{4}{c||}{NIOU-C} & \multicolumn{4}{c|}{NIOU-C:E}\\\hline
            $n$ & $(R_1,R_2)$ & MCB Cov. & $\Pr\{i^c \in \widehat{I}_n\}$  & $\widehat\E|\widehat{I}_n|$ &  $\widehat\E[\widehat{D}^+_{i^c,n}-\widehat{D}^-_{i^c,n}]$ & MCB Cov. & $\Pr\{i^c \in \widehat{I}_n\}$ & $\widehat\E|\widehat{I}_n|$ & $\widehat\E[\widehat{D}^+_{i^c,n} -\widehat{D}^-_{i^c,n}]$ 
            \\\hline\hline
             \multirow{5}{*}{100}
             & (100, 25) & 0.691 & 0.810 & 2.339 & 578.4 & 0.613 & 0.754 & 2.133 & 463.6 \\
             & (100, 100) & 0.748 & 0.863 & 2.441 & 569.2 & 0.685 & 0.809 & 2.257 & 451.7 \\
             & (400, 25) & 0.855 & 0.927 & 2.680 & 795.1 &0.787 & 0.881 & 2.501 & 613.9 \\
             & (400, 100) & 0.906 & 0.950 & 2.755 & 788.9 & 0.825 & 0.904 & 2.567 & 606.0 \\\cline{2-10}
             & Benchmark & 0.976 & 0.988 & 2.910 & 955.4 & 0.936 & 0.968 & 2.780 & 721.6 \\\hline\hline
             \multirow{5}{*}{400}
             & (400, 100) & 0.713 & 0.862 & 2.203 & 269.0 & 0.646 & 0.828 & 2.020 & 216.6 \\
             & (400, 400) & 0.781 & 0.909 & 2.274 & 263.2 & 0.688 & 0.868 & 2.067 & 209.8 \\
             & (1600, 100) & 0.864 & 0.941 & 2.525 & 371.1 & 0.800 & 0.908 & 2.351 & 288.0 \\
             & (1600, 400) & 0.902 & 0.967 & 2.594 & 369.9 & 0.851 & 0.952 & 2.417 & 283.7 \\\cline{2-10}
             & Benchmark &  0.968 & 0.990 & 2.785 & 451.3 & 0.932 & 0.978 & 2.622 & 341.5 \\\hline\hline
             \multirow{5}{*}{1600}
             & (1600, 400) & 0.721 & 0.913 & 1.837 & 122.6 & 0.641 & 0.883 & 1.716 & 97.3 \\
             & (1600, 1600) & 0.785 & 0.937 & 1.857 & 117.3 & 0.711 & 0.917 & 1.724 & 91.7 \\
             & (6400, 400) & 0.873 & 0.968 & 2.125 & 172.0 & 0.808 & 0.947 & 1.950 & 131.9 \\
             & (6400, 1600) & 0.916 & 0.972 & 2.116 & 168.3 & 0.857 & 0.957 & 1.961 & 127.0 \\\cline{2-10}
             & Benchmark & 0.967 & 0.991 & 2.315 & 209.5 & 0.936 & 0.983 & 2.122 & 155.5  \\\hline
        \end{tabular}
    }}
\end{table}

Let us first observe the benchmark results in \Cref{tab:case2}. When $n=100$, the average size of $\widehat{I}_n$ is 2.910 for NIOU-C, indicating that the procedure often includes all three solutions. The average MCB CI width for $i^c$ is 955, much greater than in Case~1, suggesting that the mean difference of the best two solutions is highly sensitive to the input distributions.
NIOU-C:E is still more aggressive than NIOU-C as in Case~1, however,  $\widehat{I}_n$ shrinks more slowly as we increase $n$, reflecting the increased difficulty of the problem.  

With simulation error, the inclusion probability of $i^c$ falls below $0.9$ for both algorithms  for smaller $n$ and $(R_1,R_2)$, but it recovers as the input and simulation data increase. A similar observation can be made for the MCB coverage while its value is lower than $\Pr\{i^c \in \widehat{I}_n\}$. 

This result combined with Case~1's suggests that NIOU-C:E may be overly aggressive when the simulation budget is limited and the problem is difficult to distinguish $i^c$ from the rest with IU.

\subsubsection{Case 3}

To study the effect of the problem size, we set $k=10$ and $\vect{a}=(64, 65, \ldots, 72, 73)$, which yields $i^c=6$. 
We generate 100 independent  samples of size $n$ from $\bF^c$ to estimate $\bFhat$, and identify the optimum under $\bFhat$; $i^c = 6$ is optimal 39\%, 68\%, and 91\% of the times for $n = 100$, $400$, and $1600$, respectively. This indicates that identifying $i^c$ in this problem is not as easy as in Case~1, but not as difficult as in Case~2.

Comparing the benchmark performance metrics of NIOU-C in \Cref{tab:case3} with those in \Cref{tab:case1,tab:case2} reveals that NIOU-C is far more conservative for large $k$, achieving a benchmark MCB coverage near or at  $1$ for all values of $n$.   This conservatism is visually supported by \Cref{fig:ellipsoid_manybox}, where the gap between $C_\alpha$ and its polyhedral superset widens with larger $k$. 
In contrast, the benchmark of NIOU-C:E does not show the same effect as it is designed to reduce the conservatism caused by the dimensionality.
As $n$ increases, both benchmarks yield  smaller $|\widehat{I}_n|$. Compared to when $k=3$ in Cases~1~and~2, the  difference in sizes of $\widehat{I}_n$ returned by NIOU-C and NIOU-C:E is more pronounced across all $n$ values. 
 
When the simulation error is introduced, 
the results show that NIOU-C achieves the target MCB coverage even with relatively small $R_1$ and $R_2$, while NIOU-C:E fails to do so, resulting in $\Pr\{i^c \in \widehat{I}_n\}$ below the target in some settings.  For larger $n$, however, NIOU-C:E delivers a smaller set while including $i^c$ in $\widehat{I}_n$ with probability above 0.9.

\begin{table}[tbp]
    \centering
    \caption{Performance comparison of NIOU-C and NIOU-C:E for Case 3 from 1000 macro-runs.\label{tab:case3}}
    {\resizebox{0.9\textwidth}{!}{%
        \begin{tabular}{|c|c||cccc||cccc|}
            \hline
            & & \multicolumn{4}{c||}{NIOU-C} & \multicolumn{4}{c|}{NIOU-C:E}\\\hline
            $n$ & $(R_1,R_2)$ & MCB Cov. & $\Pr\{i^c \in \widehat{I}_n\}$  & $\widehat\E|\widehat{I}_n|$ &  $\widehat\E[\widehat{D}^+_{i^c,n}-\widehat{D}^-_{i^c,n}]$ & MCB Cov. & $\Pr\{i^c \in \widehat{I}_n\}$ & $\widehat\E|\widehat{I}_n|$ & $\widehat\E[\widehat{D}^+_{i^c,n} -\widehat{D}^-_{i^c,n}]$
            \\\hline\hline
             \multirow{5}{*}{100}
             & (100, 25) & 0.941 & 0.978 & 5.639 & 8.7 & 0.570 & 0.769 & 2.973 & 3.1
 \\
             & (100, 100) & 0.966 & 0.993 & 5.690 & 8.4 & 0.628 & 0.830 & 2.990 & 2.9
 \\
             & (400, 25) & 0.994 & 0.997 & 7.255 & 13.7 & 0.737 & 0.890 & 3.826 & 4.5
 \\
             & (400, 100) & 0.996 & 0.999 & 7.326 & 13.3 & 0.792 & 0.924 & 3.787 & 4.2
 \\\cline{2-10}
             & Benchmark & 0.999 & 1.000 & 8.488 & 17.3 & 0.907 & 0.968 & 4.422 & 5.4
\\\hline\hline
             \multirow{5}{*}{400}
             & (400, 100) & 0.924 & 0.991 & 3.300 & 2.8 & 0.589 & 0.890 & 1.991 & 1.2
 \\
             & (400, 400) & 0.959 & 0.994 & 3.287 & 2.8 & 0.640 & 0.923 & 1.986 & 1.1
 \\
             & (1600, 100) & 0.988 & 0.997 & 4.264 & 4.4 & 0.753 & 0.938 & 2.409 & 1.7
 \\
             & (1600, 400) & 0.996 & 0.999 & 4.258 & 4.3 & 0.813 & 0.976 & 2.412 & 1.6
 \\\cline{2-10}
             & Benchmark &  0.999 & 1.000 & 4.990 & 5.6 & 0.908 & 0.989 & 2.702 & 2.0
 \\\hline\hline
             \multirow{5}{*}{1600}
             & (1600, 400) & 0.947 & 0.999 & 2.149 & 1.3 & 0.597 & 0.979 & 1.452 & 0.6
 \\
             & (1600, 1600) & 0.971 & 1.000 & 2.189 & 1.3 & 0.659 & 0.990 & 1.392 & 0.6
 \\
             & (6400, 400) & 0.995 & 1.000 & 2.673 & 1.8 & 0.784 & 0.990 & 1.678 & 0.8
 \\
             & (6400, 1600) & 0.997 & 1.000 & 2.727 & 1.8 & 0.811 & 0.993 & 1.622 & 0.8
 \\\cline{2-10}
             & Benchmark &  1.000 & 1.000 & 2.987 & 2.2 & 0.923 & 0.999 & 1.809 & 1.0
 \\\hline
        \end{tabular}
    }}
\end{table}

\subsection{Queueing System}\label{subsec:empirical-queue-system}

We compare  IOU-C and NIOU-C using a tandem queue simulation with three first-in-first-out servers. Arrivals follow a Poisson process with $\lambda=6.67$, and independent mean service times  are $\mu_1=0.73$, $\mu_2=0.7$, and $\mu_3 = 0.8$.
The first server has infinite queue capacity; the second and third have capacities of 2 and 3, introducing potential blocking.
Each server starts with a base capacity of 4, and extra capacity can be added at costs $\vect{c} = (2,5,6)$ within a budget of 9, producing $k=9$ feasible solutions. 
The objective is to minimize the average total waiting time of the first 100 customers. 
We estimate  service time distributions from data. For each macro-run and  $p \in [3]$, we generate i.i.d.\ samples of size $n_p=n \in \{100,400\}$  from the true service time distribution. IOU-C assumes exponential distributions and obtains rates via maximum likelihood estimation (MLE).
We chose $1 - \alpha = 0.9$, and NIOU-C and IOU-C use CRNs. Each replication generates $100$ variates per input, as each customer visits all three servers, i.e., $T_{ip}=100$ for $i \in [k]$, $p \in [3]$.

We test two  IOU-C variants: \emph{IOU-C: All-in} (IOU-C:A) and \emph{IOU-C: Plug-in} (IOU-C:P). Both fit a linear regression model for $\eta_i$ in the parameter vector of the input models, but IOU-C:A accounts for coefficient estimation error in MCB CIs, making it more conservative and yielding larger $\widehat{I}_n$ than IOU-C:P.

To ensure a fair comparison, all methods use the same total simulation runs. For IOU-C, \cite{Song2019} recommend running replications proportional to $n$ at  each solution with $\bFhat$ as the input model and sample $D\propto n^{1.1}$ design points from the MLEs' asymptotic distribution, with one replication per point to fit the regression model. Following their guidance, we run $4n$ replications of each solution with $\bFhat$ and $D = 4\lceil n^{1.1}\rceil$.
The total simulations per solution is $R=D+4n$, which gives $kR$ overall.  NIOU-C uses the same budget, split evenly:
$R/2$ to estimate the influence functions and $R/2$ to estimate the bounds $\set{\widehat{U}_{i\ell,n}}_{i \neq \ell}$.

\begin{table}[t]
\centering
    \caption{Comparison of IOU-C and NIOU-C procedures when service times are known to follow exponential distributions. Results are computed from $1000$ macro-runs. Each row presents  $\Pr\{i \in \widehat{I}_n\}$ for the corresponding $i$; $i^c$ is marked bold.\label{tab:queue-exp-results}}
{\resizebox{0.9\textwidth}{!}{%
\begin{tabular}{|c|c|c|cccc|cccc|}
\hline
\multirow{2}{*}{$i$} & \multirow{2}{*}{\textbf{Solution}} & \multirow{2}{*}{$\eta_i(\bF^c)$} & \multicolumn{4}{c|}{$n=100, R=1036$} & \multicolumn{4}{c|}{$n=400, R=4516$}\\
\cline{4-11}
& & & IOU-C:A & IOU-C:P & {NIOU-C} & {NIOU-C:E} & IOU-C:A & IOU-C:P & {NIOU-C} & {NIOU-C:E} \\
\hline
1 & (0, 0, 0) & 3.73 & 0.005 & 0.001 & 0 & 0 & 0 & 0 & 0 & 0 \\ 
2 & (0, 0, 1) & 3.06 & 0.543 & 0.544 & 0 & 0 & 0.049 & 0.028 & 0 & 0 \\ 
3 & (0, 1, 0) & 3.28 & 0.406 & 0.139 & 0 & 0 & 0.035 & 0.004 & 0 & 0 \\ 
4 & (1, 0, 0) & 3.25 & 0.140 & 0.002 & 0 & 0 & 0 & 0 & 0 & 0 \\ 
\textbf{5} & \textbf{(1, 0, 1)} & \textbf{2.36} & \textbf{0.992} & \textbf{0.903} & \textbf{0.997} & \textbf{0.968} & \textbf{0.996} & \textbf{0.959} & \textbf{1.000} & \textbf{0.989} \\ 
6 & (1, 1, 0) & 2.76 & 0.861 & 0.491 & 0 & 0 & 0.323 & 0.228 & 0 & 0 \\ 
7 & (2, 0, 0) & 3.12 & 0.439 & 0.001 & 0 & 0 & 0.001 & 0 & 0 & 0 \\ 
8 & (2, 1, 0) & 2.61 & 0.964 & 0.591 & 0.975 & 0.860 & 0.898 & 0.431 & 0.944 & 0.743 \\ 
9 & (3, 0, 0) & 3.08 & 0.551 & 0.062 & 0.063 & 0.019 & 0.026 & 0 & 0 & 0 \\ 

\hline
\multicolumn{3}{|r|}{MCB Coverage} & 0.965 & 0.752 & 0.987 & 0.874 & 0.965 & 0.749 & 0.987 & 0.854
 \\
\multicolumn{3}{|r|}{$\widehat\E|\widehat{I}_n|$} & 4.901 & 2.734 & 2.035 & 1.847 & 2.328 & 1.650 & 1.944 & 1.732
 \\
\hline
\end{tabular}
}}
\end{table}

In the first scenario, all three service times follow exponential distributions, and IOU-C correctly identifies the family.
\Cref{tab:queue-exp-results} shows the performance metrics averaged over 1000 macro-runs of all four methods with $n=100$ and $n=400$. Each row shows the probability of including the corresponding solution in the subset $\widehat{I}_n$ for each algorithm,
with the row for $i^c=5$ highlighted in bold. 
For $n=100$, all four algorithms include $i^c$ in  $\widehat{I}_n$ 90\% of the cases or more. However, IOU-C:A and IOU-C:P include suboptimal solutions more frequently than NIOU-C \& NIOU-C:E. The latter two exclude all suboptimal solutions from $\widehat{I}_n$ except for $i=8$ and $9$ in all $1000$ macro-runs.
While IOU-C:A has the MCB coverage of 0.965, IOU-C:P is more aggressive at 0.752. 
Both NIOU-C and NIOU-C:E have higher MCB coverage than IOU-C:P, with values of 0.987 and 0.874, respectively, the latter being the closest to the target among the four methods. 
From the results in Section~\ref{subsec:empirical-sens-analysis}, we can deduce that increasing $R$ can improve MCB coverage of NIOU-C:E to attain the target. 
For $n=400$, all procedures yield higher $\Pr\{i^c \in \widehat{I}_n\}$ and smaller $\widehat{I}_n$. Notably, both NIOU-C and NIOU-C:E have higher MCB coverage and $\Pr\{i^c \in \widehat{I}_n\}$ 
than IOU-C:P despite not leveraging the knowledge of the parametric distribution family.

\begin{table}[t]
\centering
    \caption{Comparison of IOU-C and NIOU-C procedures when service times follow unknown bimodal distributions. Results are computed from $1000$ macro-runs. Each row presents  $\Pr\{i \in \widehat{I}_n\}$ for the corresponding $i$; $i^c$ is marked bold.\label{tab:queue-bim-results}}
{\resizebox{0.9\textwidth}{!}{%
\begin{tabular}{|c|c|c|cccc|cccc|}
\hline
\multirow{2}{*}{$i$} & \multirow{2}{*}{\textbf{Solution}} & \multirow{2}{*}{$\eta_i(\bF^c)$} & \multicolumn{4}{c|}{$n=100, R=1036$} & \multicolumn{4}{c|}{$n=400, R=4519$}\\
\cline{4-11}
& & & IOU-C:A & IOU-C:P & {NIOU-C} & {NIOU-C:E} & IOU-C:A & IOU-C:P & {NIOU-C} & {NIOU-C:E} \\ \hline
1 & (0, 0, 0) & 3.64 & 0.010 & 0 & 0 & 0 & 0 & 0 & 0 & 0 \\ 
2 & (0, 0, 1) & 3.26 & 0.530 & 0.552 & 0 & 0 & 0.056 & 0.032 & 0 & 0 \\ 
3 & (0, 1, 0) & 3.12 & 0.445 & 0.159 & 0 & 0 & 0.030 & 0 & 0 & 0 \\ 
4 & (1, 0, 0) & 2.98 & 0.186 & 0.004 & 0 & 0 & 0 & 0 & 0 & 0 \\ 
5 & (1, 0, 1) & 2.38 & 1.000 & 0.972 & 0.995 & 0.933 & 1.000 & 0.989 & 0.994 & 0.933 \\ 
6 & (1, 1, 0) & 2.47 & 0.874 & 0.513 & 0 & 0 & 0.341 & 0.192 & 0 & 0 \\ 
7 & (2, 0, 0) & 2.77 & 0.483 & 0.003 & 0 & 0 & 0.009 & 0 & 0 & 0 \\ 
\textbf{8} & \textbf{(2, 1, 0)} & \textbf{2.27} & \textbf{0.985} & \textbf{0.627} & \textbf{0.993} & \textbf{0.927} & \textbf{0.970} & \textbf{0.446} & \textbf{0.994} & \textbf{0.931} \\ 
9 & (3, 0, 0) & 2.70 & 0.605 & 0.069 & 0.222 & 0.069 & 0.034 & 0 & 0.001 & 0 \\ 
\hline
\multicolumn{3}{|r|}{MCB Coverage} & 0.939 & 0.447 & 0.971 & 0.822 & 0.844 & 0.188 & 0.976 & 0.802
\\
\multicolumn{3}{|r|}{$\widehat\E|\widehat{I}_n|$} & 5.118 & 2.899 & 2.210 & 1.929 & 2.440 & 1.659 & 1.989 & 1.864
 \\
\hline
\end{tabular}
}}
\end{table}

The next scenario compares  IOU-C and NIOU-C when the former  incorrectly assumes the bimodal service time distributions to be exponential. For the $s$th server, the service time follows
$
\one{\xi=1} \cdot b_1^s \cdot \mathrm{Beta}(b_2^s, b_3^s) + \one{\xi=0} \cdot b_4^s \cdot \mathrm{Beta}(b_5^s, b_6^s),$ where $\xi \sim \mathrm{Bernoulli(\gamma_s)}$, $\vect{b}^1 = (1, 2, 6, 3, 10, 2)$, $\vect{b}^2 = (1, 2, 6, 2.3, 6, 2)$, $\vect{b}^3 = (1, 2, 6, 1, 12, 2)$, and $\vect{\gamma} = (0.785, 0.7, 0.13)$. The mean service times are  close to the first scenario, but  $i^c=8$. Moreover, this scenario is more difficult, as the gap between the optimum and the second best is smaller.
\Cref{tab:queue-bim-results} presents the results for this scenario.  For $n=100$, the probability of including  $i^c$ is above 0.9 for all algorithms other than IOU-C:P, which drops to $0.627$. In fact, IOU-C:A and IOU-C:P include $i=5$ more frequently (with statistical significance) in $\widehat{I}_n$, presuming that the problem configuration is closer to that in~\Cref{tab:queue-exp-results}.
The MCB coverage of IOU-C:P is clearly below 0.9, while NIOU-C:E also undershoots as observed in \Cref{tab:queue-exp-results}.
Although the MCB coverage of IOU-C:A meets the target, this does not mean that it is robust to an incorrect parametric assumption. Rather, it is caused by its inherent conservatism as can be seen from large $\widehat\E|\widehat{I}_n|$. 
Overall, NIOU-C shows the best performance in all metrics. 
For $n=400$, NIOU-C and NIOU-C:E excludes $i=9$ from $\widehat{I}_n$ most of the time while keeping $i=5$ and $i^c=8$ in $\widehat{I}_n$.  IOU-C:P's performance further degrades.  IOU-C:A continues to include $i^c$ in $\widehat{I}_n$ with probability 0.970, however, NIOU-C performs better than IOU-C:A by delivering a higher MCB coverage with smaller $\widehat{I}_n$.
Noticeably, the MCB coverage of IOU-C:A drops below 0.9 hinting that even with the conservatism, the incorrect parametric assumption starts to hurt its performance as IU is reduced with larger $n$. To further confirm this intuition, we increase $n$ to $1600$ and ran IOU-C:A for $1000$ macro-runs and confirm that the MCB coverage drops to 0.464.

In summary, the first scenario shows that NIOU-C performs competitively across all metrics without requiring knowledge of the correct parametric distribution families. The second scenario highlights NIOU-C’s advantage in avoiding parametric assumptions.

\section{Conclusion}
\label{sec:conclusion}

We propose the NIOU-C framework to account for IU in the R\&S problem under limited input data. 
NIOU-C constructs a confidence set of solutions with an asymptotic guarantee that the optimum is contained in the set with the user-defined confidence level. We  provide guidance on the simulation sample sizes to facilitate the asymptotic consistency and demonstrate that employing  CRNs can significantly improve the finite-sample performance.  
To mitigate the conservatism exhibited by NIOU-C in large-scale problems, we propose an alternative procedure,  NIOU-C:E\@. 
Empirical results suggest that both NIOU-C and NIOU-C:E correctly include the optimum in the confidence set  with adequate simulation effort and show robust performance compared to methods that rely on parametric assumptions.

Several extensions remain open. First, considering performance measures other than simulation output means such as quantiles or conditional value-at-risk could offer a more comprehensive risk assessment introduced by IU\@. A key challenge in this direction is the development of consistent influence function estimators tailored to these alternative risk metrics.
Extending the nonparametric treatment to accommodate more complex input distributions, such as time-varying arrival processes, could be another research direction.
Lastly, the NIOU-C framework may be extended to the case with streaming input data, where the decision maker can leverage the confidence set to stop additional data collection and return the optimum.

\bibliographystyle{informs2014} 
\bibliography{IUbib} 

\clearpage 

\appendix

\section{Proof for \Cref{sec:problem_desc}}

\begin{proof}{Proof of \Cref{lem:asymptotic-mcb-ci}}
    In the proof of Theorem~1.1 of \cite{chang1992optimal} it is shown that, if $\widehat{\eta}_i$ is a point estimator of $\eta_i(\bF^c)$ for all $i \in [k]$, then $E_1 \subseteq E_2$, where $E_1 \triangleq \{\widehat{\eta}_{i^c} - \widehat{\eta}_\ell - (\eta_{i^c}(\bF^c) - \eta_\ell(\bF^c)) \geq -q_{i\ell}, \forall \ell \neq i\} $, $E_2 \triangleq \{ \eta_i(\bF^c) - 
    \max\nolimits_{\ell  \neq i} {\eta_\ell(\bF^c)} \in 
    [D_{i}^-, D_{i}^+], \forall  i \in [k] \}$, and
    $D_{i}^+ \triangleq \big(\min_{\ell  \neq i}\set{\widehat{\eta}_{i} - \widehat{\eta}_\ell + q_{i\ell}}\big)^+$.
    Moreover
    $I \triangleq \set{i \in [k]: D_{i}^+ > 0}$,  
    $D_{i}^- = 0$ if $I = \set{i}$, and 
    $D_{i}^- = -\big(\min_{\ell  \in I: \ell \neq i}\set{\widehat{\eta}_{i} - \widehat{\eta}_\ell - q_{\ell i}}\big)^-$ otherwise, where $\{q_{i\ell}\}_{i\neq\ell}$ are fixed quantities.
    We observe that $E_1 = \{\eta_{i^c}(\bF^c) - \eta_\ell(\bF^c) \leq \widehat{\eta}_{i^c} - \widehat{\eta}_\ell + q_{i\ell}, \forall \ell \neq i\}$. 
    Below, we apply this result to our framework.
    
    If we denote $U_{i\ell,n} = \widehat{\eta}_{i} - \widehat{\eta}_\ell + q_{i\ell}$, then $E_1^n \triangleq \{\eta_{i^c}(\bF^c) - \eta_\ell(\bF^c) \leq 
    U_{i^c\ell,n}, \forall \ell \neq i\}$ and $E_2^n \triangleq \{ \eta_i(\bF^c) - 
    \max\nolimits_{\ell  \neq i} {\eta_\ell(\bF^c)} \in 
    [D_{i,n}^-, D_{i,n}^+], \forall  i \in [k] \}$. Therefore $E_1^n \subseteq E_2^n$. The lemma's hypotheses ensure $1-\alpha \leq \liminf_{n\to\infty} \Pr\{E_1^n\}$. Hence  $1-\alpha \leq \liminf_{n\to\infty}\Pr\{E_2^n\}$.
    The event $E_2^n$ implies that no solution satisfying $D_{i,n}^+ = 0$ is optimal, yielding $i^c \in I_n$. Hence $\Pr\{E_2^n\} \leq \Pr\{i^c \in I_n\}$.
    Taking limits, we obtain $1-\alpha \leq \liminf_{n\to\infty} \Pr\{i^c \in I_n\}$.
    
\end{proof}

\section{Proofs for \Cref{sec:as}}

We first establish properties of the optimization problem in \eqref{eq:R}. First, $-\log R(\vect{0})$ is equal to the optimal value of
\begin{equation}\label{eq:pf-opt}
    \min_{\substack{\vect{w}_p, \, \bmu_p \\ p \in [m]}}  \;  -\sum\nolimits_{p=1}^m \sum\nolimits_{j=1}^{n_p} \log(n_pw_{pj})  \quad 
    \text{s.t.} \quad  \sum\nolimits_{j=1}^{n_p} \vect{Y}_{pj} w_{pj} = \bmu_p, \;\sum\nolimits_{j=1}^{n_p} w_{pj} = 1, \; p\in [m], \; 
    \sum\nolimits_{p=1}^{m} \bmu_p = \vect{0}.
\end{equation}
In \eqref{eq:pf-opt}, as opposed to \eqref{eq:R}, we introduce slack variables $\bmu_p \in \mathbb{R}^q$, and drop the nonnegativity constraints on $\vect{w}_p$  because they are implicitly imposed in the objective function.

Let us introduce some notation. 
Let $\{\vect{w}^*_{p}\}_{p \in [m]}$ and $\{\bmu_p^*\}_{p \in [m]}$ be minimizers of \eqref{eq:pf-opt}, and let $\lambst$ be the Lagrangian multiplier associated with the constraint $\sum\nolimits_{p=1}^m \bmu_p = \vect{0}$.
Let $\overline{\vect{Y}}_{p} \triangleq (1/n_p)\sum\nolimits_{j=1}^{n_p} \vect{Y}_{pj}$ be the sample mean of the observations of each input distribution $p \in [m]$. In addition, we define $v_{pj} \triangleq \big( \lambst/n_p \big) ^\top(\vect{Y}_{pj} - \bmu_p^*)$, $ \matr{S}_p \triangleq (1/n_p) \sum\nolimits_{j=1}^{n_p} (\vect{Y}_{pj} - \bmu_p^*)(\vect{Y}_{pj} - \bmu_p^*)^\top$ and $\vect{\nu}_p \triangleq (1/n_p) \sum\nolimits_{j=1}^{n_p} \frac{(\vect{Y}_{pj} - \bmu_p^*) v_{pj}^2}{1 + v_{pj}}$.
To compute $-\log R(\vect{0})$, we use the following proposition that characterizes minimizers of \eqref{eq:pf-opt}, along with the convergence rate of the Lagrangian multipliers.
\begin{proposition}\label{prop:opt_solution}
    Each minimizer  $\bw^* \triangleq \{\vect{w}^*_{p}\}_{p \in [m]}$ and $\bmu^* \triangleq \{\bmu_p^*\}_{p \in [m]}$
    of \eqref{eq:pf-opt} satisfies $w_{pj}^* = (1/n_p) \cdot 1/(1 + v_{pj})$. 
    In addition,
    $\bmu_p^* \xrightarrow{P} \vect{0}$, $\norm{\lambst}/n = O_p(n^{-1/2})$,
    $\lambst = \big( \sum\nolimits_{p=1}^{m} (\matr{S}_p/n_p) \big)^{-1} \sum\nolimits_{p=1}^{m} (\overline{\vect{Y}}_{p} + \vect{\nu}_p)$,
    and $\vect{\nu}_p = o_p(n^{-1/2})$
    as $n \to\infty$.
\end{proposition}

\begin{proof}{Proof}
Let $\vect{w}^*\triangleq \{w^*_{pj}\}_{j \in [n_p], p \in [m]}$ and $\bmu^* \triangleq \{\bmu_p^*\}_{p \in [m]}$ be the minimizers of \eqref{eq:pf-opt}, and $\{\lambst_{1p}, \lambda_{2p}^*\}_{p \in [m]}$ and $\lambst$ be the corresponding Lagrangian multipliers.
This proof consists of four steps. In Step 1, we obtain an expression for $\vect{w}^*$. In Step 2, we establish the convergence rate of the Lagrangian multiplier $\lambst$, while in Step 3, we determine the convergence rate of $\bmu_p^*$ to the true mean of each input distribution $p$. Lastly, in Step 4, we obtain an expression for $\lambst$.

\emph{Step 1}: Since the problem in \eqref{eq:pf-opt} has linear constraints and the objective function is convex, the Karush--Kuhn--Tucker (KKT) conditions are necessary and sufficient optimality conditions. The KKT conditions  are given by:

\begin{scriptsize}
\begin{equation}\label{eq:kkt}
     - \frac{1}{w^*_{pj}} + (\lambst_{1p})^\top\vect{Y}_{pj} + \lambda^*_{2p} = 0,\,   j \in [n_p]; \,
     -\lambst_{1p} + \lambst  = \vect{0}, \,
     \sum\nolimits_{j=1}^{n_p} \vect{Y}_{pj} w^*_{pj} = \bmu^*_p, \,
     \sum\nolimits_{j=1}^{n_p} w^*_{pj}  = 1, \, p \in [m];
     \sum\nolimits_{p=1}^{m} \bmu^*_p  = \vect{0},
\end{equation}%
\end{scriptsize}%
where $(\lambda^*_{1p},\lambda^*_{2p})$ are the  multipliers corresponding to the first two constraints of~\eqref{eq:pf-opt}, and $\lambst$ is that of the last constraint.

The values for $\vect{w}^*, \bmu^*$ and $\lambst$ can be obtained by solving \eqref{eq:kkt}.
From the second equation in \eqref{eq:kkt}, $\lambst_{1p} = \lambst$ and plugging this into the first KKT condition and 
multiplying both sides by $w^*_{pj}$ gives
\begin{equation} \label{eq:kkt.21}
    \lambda_{2p}^*w^*_{pj} = 1 - (\lambst)^\top\vect{Y}_{pj}w^*_{pj}, 
    \quad p \in [m],j \in [n_p].
\end{equation}
Summing~\eqref{eq:kkt.21} over all $j \in [n_p]$ and using the fourth KKT equation, we obtain
$\lambda_{2p}^* = n_p - (\lambst)^\top\bmu_p^*$, $p \in [m]$.
We plug $\lambst_{1p}$ and $\lambda_{2p}^*$ into the first KKT condition
to obtain
$w_{pj}^*  = 1/\big(n_p + (\lambst)^\top(\vect{Y}_{pj} - \bmu_p^*)\big)
= (1/n_p) \cdot 1/(1 + v_{pj}).$

\emph{Step 2:} We characterize the convergence rate of $\norm{\lambst}/n$ as $n \to \infty$.
Let $\Theta$ be the set of all unit vectors in $\mathbb{R}^q$.
Let $\lambst = \norm{\lambst} \vect{\theta}$ for some
$\vect{\theta} \in \Theta$.
Note that we have $w_{pj}^* = (1/n_p)\cdot1/(1+v_{pj})$ and $1/(1+v_{pj}) = 1 - v_{pj}/(1+v_{pj})$. 
Inserting $w_{pj}^*$ into the third KKT condition and using the definition of $v_{pj}$, we get
\begin{equation}\label{eq:g_lambda_zero}
    \frac{1}{n_p} \sum\nolimits_{j=1}^{n_p} \frac{\vect{Y}_{pj} - \bmu_p^*}{1 + (\frac{\lambst}{n_p})^\top(\vect{Y}_{pj} - \bmu_p^*)} = \vect{0}.
\end{equation}
Rewriting \eqref{eq:g_lambda_zero} yields
\begin{equation*}
    \overline{\vect{Y}}_{p} - \bmu_p^* = \frac{1}{n_p} \sum\nolimits_{j=1}^{n_p} \frac{(\vect{Y}_{pj} - \bmu_p^*)(\vect{Y}_{pj} - \bmu_p^*)^\top}{1 + v_{pj}} \bigg( \frac{\lambst}{n_p} \bigg).
\end{equation*} 
Taking the inner product with $\vect{\theta}$ and defining $\widetilde{\matr{S}}_p \triangleq (1/n_p) \sum\nolimits_{j=1}^{n_p} \frac{(\vect{Y}_{pj} - \bmu_p^*)(\vect{Y}_{pj} - \bmu_p^*)^\top}{1 + v_{pj}}$, we have
\begin{equation}\label{eq:theta_both_sides}
    \vect{\theta}^\top(\overline{\vect{Y}}_{p} - \bmu_p^*) = \frac{\norm{\lambst}}{n_p} \vect{\theta}^\top \widetilde{\matr{S}}_p \vect{\theta}.
\end{equation}

Let $K^* \triangleq \max_{p \in [m], j \in [n_p]}\set{\norm{\vect{Y}_{pj} - \bmu_p^*}}$. Then,
\begin{equation*}
     \frac{\norm{\lambst}}{n_p} \vect{\theta}^\top \matr{S}_p \vect{\theta} \leq  \frac{\norm{\lambst}}{n_p} \vect{\theta}^\top \widetilde{\matr{S}}_p \vect{\theta} (1 + \max_{j \in [n_p]}|v_{pj}|)
     \leq \frac{\norm{\lambst}}{n_p} \vect{\theta}^\top \widetilde{\matr{S}}_p \vect{\theta} \bigg(1 + \frac{\norm{\lambst}}{n_p} K^*\bigg)
     = \vect{\theta}^\top(\overline{\vect{Y}}_{p} - \bmu_p^*) \bigg(1 + \frac{\norm{\lambst}}{n_p} K^*\bigg).
\end{equation*}
We have $(\norm{\lambst}/n_p) \, \vect{\theta}^\top \matr{S}_p \vect{\theta} \geq (\norm{\lambst}/\overline{c}n) \, \vect{\theta}^\top \matr{S}_p \vect{\theta}$ for sufficiently large $n$. From $n_p \geq \underline{c}n$, we also have
$
    \frac{\norm{\lambst}}{\overline{c}n} \vect{\theta}^\top \matr{S}_p \vect{\theta} \leq \vect{\theta}^\top(\overline{\vect{Y}}_{p} - \bmu_p^*) \Big(1 + \frac{\norm{\lambst}}{\underline{c}n} K^*\Big),
$ 
which can be rewritten as
$
    \frac{\norm{\lambst}}{\overline{c}n} \bigg[ \vect{\theta}^\top \matr{S}_p \vect{\theta} - \frac{\overline{c}}{\underline{c}} K^* \vect{\theta}^\top(\overline{\vect{Y}}_{p} - \bmu_p^*) \bigg]  \leq \vect{\theta}^\top(\overline{\vect{Y}}_{p} - \bmu_p^*).
$ 
Summing over $p \in [m]$ and using the last equation in \eqref{eq:kkt} yields
\begin{equation}
\label{eq:lambda-order-1}
    \frac{\norm{\lambst}}{\overline{c}n} \Bigg[ \vect{\theta}^\top \bigg(\sum\nolimits_{p=1}^{m} \matr{S}_p\bigg) \vect{\theta} - \frac{\overline{c}}{\underline{c}} K^* \vect{\theta}^\top \bigg(\sum\nolimits_{p=1}^{m} \overline{\vect{Y}}_{p}\bigg) \Bigg]  \leq \vect{\theta}^\top \bigg(\sum\nolimits_{p=1}^{m} \overline{\vect{Y}}_{p}\bigg).
\end{equation}

Observe that, for each $p \in [m]$,
$
    \vect{\theta}^\top \matr{S}_p \vect{\theta} = \vect{\theta}^\top \big(\widehat{\matr{V}}_p - 2 \bmu_p^* (\overline{\vect{Y}}_{p})^\top + \bmu_p^* {\bmu_p^*}^\top \big) \vect{\theta} \geq \vect{\theta}^\top \big(\widehat{\matr{V}}_p - 2 \bmu_p^* (\overline{\vect{Y}}_{p})^\top \big) \vect{\theta}
$, 
where $\widehat{\matr{V}}_p \triangleq (1/n_p) \sum\nolimits_{j=1}^{n_p} \vect{Y}_{pj} \vect{Y}_{pj}^\top$.
Therefore, 
\begin{equation*}
    \frac{\norm{\lambst}}{\overline{c}n} \Bigg[ \vect{\theta}^\top \bigg(\sum\nolimits_{p=1}^{m}\widehat{\matr{V}}_p \bigg) \vect{\theta} - 2 \vect{\theta}^\top \bigg( \sum\nolimits_{p=1}^{m} \bmu_p^* (\overline{\vect{Y}}_{p})^\top \bigg) \vect{\theta} - \frac{\overline{c}}{\underline{c}} K^* \vect{\theta}^\top \bigg(\sum\nolimits_{p=1}^{m} \overline{\vect{Y}}_{p} \bigg) \Bigg] \leq \vect{\theta}^\top \bigg(\sum\nolimits_{p=1}^{m} \overline{\vect{Y}}_{p}\bigg).
\end{equation*}

By the law of large numbers and $\E[\vect{Y}_p] = \vect{0}$, we get
$\widehat{\matr{V}}_p \xrightarrow{\text{a.s.}} \matr{V}_p$. Also, since $m$ is finite and $\norm{\vect{\theta}}=1$,
$\sum\nolimits_{p=1}^{m}\tau^p_q + o_p(1) \leq \vect{\theta}^\top  \big(\sum\nolimits_{p=1}^{m} \widehat{\matr{V}}_p \big) \vect{\theta} \leq \sum\nolimits_{p=1}^{m} \tau^p_1 + o_p(1)$,
where $\tau^p_1$ and $\tau^p_q$ are the largest and smallest eigenvalues of $\matr{V}_p$. We also have $K^* = o(n^{1/2})$ almost surely as a result of Lemma~11.2 in \cite{owenEL}. In addition, a CLT applied to each $\overline{\vect{Y}}_{p}$ implies that $\vect{\theta}^\top \big(\sum\nolimits_{p=1}^{m} \overline{\vect{Y}}_{p} \big) = O_p(n_p^{-1/2}) = O_p(n^{-1/2})$. Also, $\norm{\bmu_p^*} 
\leq \max_{p\in [m],j\in [n_p]} \norm{\vect{Y}_{pj}} = o(n^{1/2})$ almost surely by Lemma~11.2 in \cite{owenEL}.
Thus,
\begin{equation*}
    \frac{\norm{\lambst}}{\overline{c}n} \bigg(\vect{\theta}^\top \Big(\sum\nolimits_{p=1}^{m} \widehat{\matr{V}}_p\Big) \vect{\theta} - o(n^{1/2}) O_p(n^{-1/2}) - o(n^{1/2}) O_p(n^{-1/2}) \bigg)  = O_p(n^{-1/2}),
\end{equation*}
and $\norm{\lambst}/n = O_p(n^{-1/2})$ follows.

\emph{Step 3:} We  show that $\bmu_p^* \xrightarrow{P} \vect{0}$ as $n\to\infty$. By the law of large numbers, $\overline{\vect{Y}}_{p} \xrightarrow{\text{a.s.}} \E[\vect{Y}_p] = \vect{0}$. Hence, it suffices to show that $\norm{\overline{\vect{Y}}_{p} - \bmu_p^*} \xrightarrow{P} \vect{0}$ as $n\to\infty$. 
Observe that $\overline{\vect{Y}}_{p} - \bmu_p^* = \sum\nolimits_{j=1}^{n_p}(1/n_p - w_{pj}^*)\vect{Y}_{pj} = (1/n_p) \sum\nolimits_{j=1}^{n_p} \big(v_{pj}/(1 + v_{pj})\big) \vect{Y}_{pj}$. We can bound $\norm{\overline{\vect{Y}}_{p} - \bmu_p^*} $ from above with:
\begin{footnotesize}
\begin{equation*}
    \frac{1}{n_p} \sum_{j=1}^{n_p} \frac{|v_{pj}|}{1 + v_{pj}} \norm{\vect{Y}_{pj}} 
    \leq \frac{1}{n_p} \sum_{j=1}^{n_p} \frac{\max_{j \in [n_p]}|v_{pj}|}{\min_{j \in [n_p]} 1 + v_{pj}} \norm{\vect{Y}_{pj}} 
      = \frac{\max_{j \in [n_p]}|v_{pj}|}{1 + \min_{j \in [n_p]} v_{pj}} \cdot \frac{1}{n_p} \sum_{j=1}^{n_p} \norm{\vect{Y}_{pj}} 
    = \frac{o_p(1) O_p(1)}{1 + o_p(1)}  = o_p(1).
\end{equation*}
\end{footnotesize}%

\emph{Step 4:} We derive $\lambst$ and show that $\vect{\nu}_p = o_p(n^{-1/2})$. Using $(1 + v_{pj})^{-1} = 1 - v_{pj} + v_{pj}^2/(1 + v_{pj})$, we obtain
\begin{equation*}
    \frac{1}{n_p} \sum\nolimits_{j=1}^{n_p} \frac{\vect{Y}_{pj} - \bmu_p^*}{1 + v_{pj}} 
    = \overline{\vect{Y}}_{p} - \bmu_p^* - \frac{1}{n_p} \sum\nolimits_{j=1}^{n_p} (\vect{Y}_{pj} - \bmu_p^*) (\vect{Y}_{pj} - \bmu_p^*)^\top \bigg(\frac{\lambst}{n_p}\bigg) + \frac{1}{n_p} \sum\nolimits_{j=1}^{n_p} \frac{(\vect{Y}_{pj} - \bmu_p^*)v_{pj}^2}{1 + v_{pj}}.
\end{equation*}
Combined with \eqref{eq:g_lambda_zero},
\begin{equation}\label{eq:g_lambda_zero_2}
    \overline{\vect{Y}}_{p} - \bmu_p^* = \matr{S}_p \bigg(\frac{\lambst}{n_p}\bigg) - \frac{1}{n_p} \sum\nolimits_{j=1}^{n_p} \frac{(\vect{Y}_{pj} - \bmu_p^*) v_{pj}^2}{1 + v_{pj}} = \bigg( \frac{\matr{S}_p}{n_p} \bigg) \lambst - \vect{\nu}_p.
\end{equation}
Then, the right-most term in \eqref{eq:g_lambda_zero_2} has a norm that can be bounded as
{
\setlength{\abovedisplayskip}{5pt}
\setlength{\belowdisplayskip}{0pt}
\begin{align*} 
    \norm{\vect{\nu}_p}
    & \leq \frac{1}{n_p} \sum\nolimits_{j=1}^{n_p} v_{pj}^2\bignorm{ \frac{(\vect{Y}_{pj} - \bmu_p^*) }{1 + v_{pj}} } 
    \leq \bigg(\frac{\norm{\lambst}}{n_p}\bigg)^2 \frac{1}{n_p} \sum\nolimits_{j=1}^{n_p} \norm{\vect{Y}_{pj} - \bmu_p^*}^3 \bigg|\frac{1}{1 + v_{pj}}\bigg| \\
    & \leq \bigg(\frac{\norm{\lambst}}{n_p}\bigg)^2 \frac{1}{\min_{j\in [n_p]} 1 + v_{pj}} \cdot \frac{1}{n_p} \sum\nolimits_{j=1}^{n_p} \norm{\vect{Y}_{pj} - \bmu_p^*}^3 
    = O_p(n^{-1}) \frac{1}{1 + o_p(1)} o(n^{1/2}) = o_p(n^{-1/2}).
\end{align*}
}%
Summing over $p \in [m]$ and rearranging terms,
\begin{equation*}
    \lambst = \bigg( \sum\nolimits_{p=1}^{m} \frac{\matr{S}_p}{n_p} \bigg)^{-1} \sum\nolimits_{p=1}^{m} \big(\overline{\vect{Y}}_{p} + \vect{\nu}_p\big).\hfill 
\end{equation*}
\end{proof}

\begin{proof}{Proof of \Cref{thm:multisample-el}}
This proof is inspired by that of Theorem~4 in \cite{lam2017optimization}.  
We have
\begin{equation*}
    R(\bar{\bmu})  = \max\Big\{\prod\nolimits_{p=1}^m\prod\nolimits_{j=1}^{n_p} n_p {w_{pj}} \colon \sum\nolimits_{p=1}^m \sum\nolimits_{j=1}^{n_p} (\vect{Y}_{pj} - \E[\vect{Y}_p]) w_{pj} = \vect{0}, 
    \quad \vect{w}_p \in \Delta_{n_p}, \quad  p \in [m]  \Big\},
\end{equation*}
where this expression is the log likelihood ratio defined for the translated observations.
Then, without loss of generality, we assume $\E[\vect{Y}_p] = \vect{0}$, so that $R(\bar{\bmu}) = R(\vect{0})$.

From \Cref{prop:opt_solution}, we have $-2\log{R(\vect{0})}  = -2 \sum\nolimits_{p=1}^{m} \sum\nolimits_{j=1}^{n_p} \log (n_p w_{pj}^*) 
    = 2 \sum\nolimits_{p=1}^{m} \sum\nolimits_{j=1}^{n_p} \log (1 + v_{pj})$.
By the mean value theorem and a second-order Taylor's expansion, for each $p \in [m]$ and $j \in [n_p]$, we have
$\log (1 + v_{pj}) = v_{pj} - (1/2) v_{pj}^2 + \phi_{pj}$,
where $\phi_{pj} \triangleq \frac{v_{pj}^3}{3(1 + \xi_{pj}v_{pj})^3}$ for some $\xi_{pj} \in (0,1)$. Then, using the definition of $v_{pj}$ and $\lambst$, 
{
\setlength{\abovedisplayskip}{5pt}
\setlength{\belowdisplayskip}{0pt}
\begin{align}
    -2\log{R(\vect{0})} 
    & =  2 (\lambst)^\top \bigg(\sum_{p=1}^{m} \overline{\vect{Y}}_{p}\bigg) - (\lambst)^\top \bigg(\sum_{p=1}^{m} \frac{\matr{S}_p}{n_p}\bigg) \lambst + 2 \sum_{p=1}^{m} \sum_{j=1}^{n_p} \phi_{pj} \notag\\
    & = \bigg(\sum_{p=1}^{m} \overline{\vect{Y}}_{p}\bigg)^\top \bigg( \sum_{p=1}^{m} \frac{\matr{S}_p}{n_p} \bigg)^{-1} \bigg(\sum_{p=1}^{m} \overline{\vect{Y}}_{p} \bigg) - \bigg(\sum_{p=1}^{m} \vect{\nu}_p\bigg)^\top  \bigg(\sum_{p=1}^{m} \frac{\matr{S}_p}{n_p}\bigg)^{-1} \bigg( \sum_{p=1}^{m} \vect{\nu}_p \bigg) + 2 \sum_{p=1}^{m} \sum_{j=1}^{n_p} \phi_{pj}. \label{eq:converge_equation_1}
\end{align}
}

Next, we show that the first term in \eqref{eq:converge_equation_1} converges to a $\chi^2_{q}$ random variable in distribution, and the remaining terms are $o_p(1)$.  To establish this for the first term in \eqref{eq:converge_equation_1}, it suffices to show that
\begin{equation}\label{eq:converge_equation_2}
    \bigg( \sum\nolimits_{p=1}^{m} \frac{\matr{S}_p}{n_p} \bigg)^{-1/2} \sum\nolimits_{p=1}^{m} \overline{\vect{Y}}_{p} \, \xrightarrow{D} \, \N(\vect{0}, \matr{I})
    \quad \text{as} \quad n\to\infty.
\end{equation}
We show this by verifying the conditions needed to apply Theorem~8.6.1 in \cite{Borovkov2013},
which presents a multivariate triangular array version of the Lindeberg--Feller CLT. Let $N \triangleq \sum\nolimits_{p=1}^m n_p$, $\overline{\matr{V}}_n \triangleq \sum\nolimits_{p=1}^{m} (1/n_p)\matr{V}_p$ and
\[ (\vect{x}_{N1}, \dots , \vect{x}_{NN}) \triangleq (\overline{\matr{V}}_n)^{-1/2}(\vect{Y}_{11}/n_1, \dots, \vect{Y}_{1n_1}/n_1, \dots, \vect{Y}_{m1}/n_m, \dots,  \vect{Y}_{mn_m}/n_m). \]
We have
$    \vect{z}_N = \sum\nolimits_{i=1}^N \vect{x}_{Ni} 
    = (\overline{\matr{V}}_n)^{-1/2} \sum\nolimits_{p=1}^{m} \overline{\vect{Y}}_{p}$ and $\matr{V}_N 
    = \matr{I}.$
From \Cref{assump:data_size} we can derive that there exists constants $0 < \underline{c}, \overline{c} < \infty$ such that $\underline{c} < n_p/n < \overline{c}$ holds for $p\in [m]$ and for all sufficiently large $n \in \mathbb{N}$. Therefore, 
using the fact that $N = n\cdot m$ and $\underline{c}N/m \leq n_p$ for all sufficiently large $n$, we verify the conditions: For each $\epsilon >0$,
{
\setlength{\abovedisplayskip}{5pt}
\setlength{\belowdisplayskip}{0pt}
\begin{align*}
    \sum\nolimits_{i=1}^N \E\Big[\norm{\vect{x}_{Ni}}^2 \one{\norm{\vect{x}_{Ni}} \geq \epsilon}\Big] & 
    = \sum\nolimits_{p=1}^{m} \sum\nolimits_{j=1}^{n_p} \E\Bigg[ \bignorm{\frac{(\overline{\matr{V}}_n)^{-1/2}\vect{Y}_{pj}}{n_p}}^2 \mathbf{1}\Big\{\big\|(\overline{\matr{V}}_n)^{-1/2}\vect{Y}_{pj}\big\| \geq \epsilon  n_p\Big\}\Bigg]
    \\ & \leq \sum\nolimits_{p=1}^{m} \frac{\big|\big|(\overline{\matr{V}}_n)^{-1/2}\big|\big|^2_F}{n_p} \E\Big[\norm{\vect{Y}_{p1}}^2 \mathbf{1}\Big\{\norm{\vect{Y}_{p1}} \geq \epsilon \big\|(\overline{\matr{V}}_n)^{-1/2}\big\|^{-1}_F n_p\Big\}\Big] \\ 
    & \leq \sum\nolimits_{p=1}^{m} C_1 \E\Big[\norm{\vect{Y}_{p1}}^2 \mathbf{1}\big\{\norm{\vect{Y}_{p1}} \geq \epsilon C_2 N^{1/2}\big\}\Big] \to 0.
\end{align*}
}%
Here,  $\norm{\cdot}_F$ is the Frobenius norm, so the first inequality comes from the inequality $\norm{(\overline{\matr{V}}_n)^{-1/2}\vect{Y}_{pj}} \leq \norm{(\overline{\matr{V}}_n)^{-1/2}}_F\norm{\vect{Y}_{pj}}$.
The second inequality results from $\overline{\matr{V}}_n = O(1/n)$, so there exist constants $C_1$ and $C_2$ such that $(1/n)\norm{(\overline{\matr{V}}_n)^{-1/2}}^2_F \leq C_1$ and $\norm{(\overline{\matr{V}}_n)^{-1/2}}^{-1}_F n_p \leq C_2 (mn)^{1/2} = C_2 N^{1/2}$ for all large enough $N$.
Finally, the convergence to zero is justified by the dominated convergence theorem since $\E\big[\norm{\vect{Y}_{p1}}^2 \one{\norm{\vect{Y}_{p1}} \geq \epsilon C_2 N^{1/2}}\big] \leq \E[\norm{\vect{Y}_{p1}}^2] < \infty$ and $\one{\norm{\vect{Y}_{p1}} \geq \epsilon C_2 N^{1/2}} \xrightarrow{P} 0$ as $N \to\infty$.
Having verified the Lindeberg--Feller conditions, and since $\matr{S}_p \xrightarrow{P} \matr{V}_p$ as $n\to\infty$ from \Cref{prop:opt_solution},
we obtain \eqref{eq:converge_equation_2}.
Consequently, the first term in \eqref{eq:converge_equation_1} converges in distribution to $\chi^2_q$.

For the second term in \eqref{eq:converge_equation_1}, first note that $(\sum\nolimits_{p=1}^{m} \matr{S}_p/n_p)^{-1} = O_p(n)$
as $n \to \infty$ because
$\big(\sum\nolimits_{p=1}^{m} \matr{S}_p/n_p\big)^{-1} 
= n \big(\sum\nolimits_{p=1}^{m} \matr{S}_p n/n_p\big)^{-1},$
 $\matr{S}_p \xrightarrow{P} \matr{V}_p$ according to \Cref{prop:opt_solution}, and $n_p/n \to \beta_p$. Hence 
\begin{equation*}
     \bigg(\sum\nolimits_{p=1}^{m} \vect{\nu}_p\bigg)^\top  \bigg(\sum\nolimits_{p=1}^{m} \frac{\matr{S}_p}{n_p}\bigg)^{-1}  \bigg(\sum\nolimits_{p=1}^{m} \vect{\nu}_p\bigg)= o_p\big(n^{-1/2}\big)O_p(n)o_p\big(n^{-1/2}\big) = o_p(1).
\end{equation*}

To bound the last term in \eqref{eq:converge_equation_1}, we first note that $(1/n_p) \sum\nolimits_{j=1}^{n_p}\norm{\vect{Y}_{pj}-\bmu_p^*}^3 = o(n^{1/2})$ almost surely by Lemma~11.3 in \cite{owenEL}. Applying Lemma~11.2 in \cite{owenEL}, we obtain $\max_{j \in [n_p]} |v_{pj}| = \max_{j \in [n_p]} |(\lambst/n_p)^\top (\vect{Y}_{pj} - \bmu_p^*)| \leq  (\norm{\lambst}/n_p) \max_{j \in [n_p]} \norm{\vect{Y}_{pj} - \bmu_p^*} = O_p(n^{-1/2})o(n^{1/2}) = o_p(1)$. Also, the optimality of $w_{pj}^*$ yields $w_{pj}^*> 0$, which implies $1 + v_{pj} > 0$. This ensures $\min_{j \in [n_p]} |1 + v_{pj}| 
= 1 + \min_{j \in [n_p]} v_{pj}$ and $|\min_{j \in [n_p]} v_{pj}| \leq \max_{j \in [n_p]} |v_{pj}| = o_p(1)$. Hence, $\min_{p,j} \xi_{pj}v_{pj} = o_p(1)$ since $\xi_{pj}\in(0,1)$. In addition, from \Cref{prop:opt_solution}, $\|\lambst/(\underline{c}n)\|^3 = O_p(n^{-3/2})$. Using these results, we obtain the bounds
{
\setlength{\abovedisplayskip}{5pt}
\setlength{\belowdisplayskip}{0pt}
\begin{align*}
    \bigg| 2 \sum\nolimits_{p=1}^{m}\sum\nolimits_{j=1}^{n_p} \phi_{pj} \bigg| & \leq  \frac{2}{3} \sum\nolimits_{p=1}^{m}\sum\nolimits_{j=1}^{n_p} \bigg| \frac{v_{pj}^3}{(1 + \xi_{pj}v_{pj})^3} \bigg| \\
    & \leq \frac{2}{3}  \bignorm{\frac{\lambst}{\underline{c}n}}^3 \frac{1}{(1 + \min_{p,j} \xi_{pj}v_{pj})^3} \sum\nolimits_{p=1}^{m} n_p \cdot \frac{1}{n_p} \sum\nolimits_{j=1}^{n_p} \norm{\vect{Y}_{pj} - \bmu_p^*}^3  \\
    & = O_p(n^{-3/2}) \frac{1}{1 + o_p(1)} O(n) o(n^{1/2}) = o_p(1).
\end{align*}
}
Combining all pieces via Slutsky's theorem, $-2\log{R(\vect{0})}  \xrightarrow{D} \chi^2_q$. \hfill
\end{proof}

In \Cref{sec:as} we introduce $\vect{B}$ defined in \eqref{eq:generic_opt}. We present a proof of \Cref{lem:mean.CR.coverage} that characterizes the asymptotic coverage of $\vect{B}$ and other results that allow us to state \Cref{prop:upper_bound.cov.opt}, which characterizes an asymptotic bound on  $\vect{B}$.

\begin{proof}{Proof of \Cref{lem:mean.CR.coverage}}
    \Cref{thm:multisample-el} implies $\Pr\{\bar{\bmu} \in C_\alpha\} \to 1 - \alpha$ as $n \to \infty$.
    Defining $C_B \triangleq\set{\bmu \in \mathbb{R}^q: \bmu \leq \vect{B}}$, we have $C_\alpha \subset C_B$. Hence $\liminf_{n \to \infty} \Pr\{\bar{\bmu} \leq \vect{B}\} \geq 1 - \alpha$. \hfill   
\end{proof}

\begin{corollary}[{of \Cref{thm:multisample-el}}]
\label{cor:thm3}
    If the hypotheses of \Cref{thm:multisample-el} are satisfied, and $\vect{z}_n \to \vect{z} \in \mathbb{R}^q$, then
    $-2 \log R\big( \sum\nolimits_{p=1}^{m} \overline{\vect{Y}}_{p} - (\sum\nolimits_{p=1}^m \matr{V}_p/n_p)^{1/2} \vect{z}_n\big) \xrightarrow{P} \norm{\vect{z}}^2$.
\end{corollary}
\begin{proof}{Proof}
    The proof follows the same structure as that of \Cref{prop:opt_solution} and \Cref{thm:multisample-el}, with only minor modifications. We highlight the necessary adjustments to the original argument to establish this corollary.
    In \eqref{eq:pf-opt}, we replace the last constraint by 
    \[ \sum\nolimits_{p=1}^{m} \bmu_p = \sum\nolimits_{p=1}^{m} \overline{\vect{Y}}_{p} - \Big(\sum\nolimits_{p=1}^m (1/n_p)\matr{V}_p\Big)^{1/2} \vect{z}_n. \] This modification also applies to the last equation of the KKT conditions. 
    
    In \emph{Step 1}, we change \eqref{eq:lambda-order-1} to
    \begin{equation*}
       \frac{\norm{\lambst}}{\overline{c}n} \bigg( \vect{\theta}^\top \Big(\sum\nolimits_{p=1}^{m} \matr{S}_p\Big) \vect{\theta} - \frac{\overline{c}}{\underline{c}} K^* \vect{\theta}^\top \Big(\sum\nolimits_{p=1}^m (1/n_p) \matr{V}_p\Big)^{1/2}\vect{z}_n \bigg)  \leq \vect{\theta}^\top  \Big(\sum\nolimits_{p=1}^m (1/n_p) \matr{V}_p\Big)^{1/2} \vect{z}_n.   
    \end{equation*}
    Since $\sum\nolimits_{p=1}^m (1/n_p)\matr{V}_p = O(n^{-1})$, we can show that $(1/n)\norm{\lambst} = O_p(n^{-1/2})$.

    In \emph{Step 4}, we obtain 
    $
        \big(\sum\nolimits_{p=1}^m (1/n_p)\matr{V}_p\big)^{1/2} \vect{z}_n = \big(\sum\nolimits_{p=1}^m (1/n_p) \matr{S}_p\big) \lambst + \sum\nolimits_{p=1}^m \vect{\nu}_p. 
    $ 
    Then, considering $\matr{S}_p = \matr{V}_p + o_p(1)$ and rearranging terms, we have
    \begin{equation}\label{eq:cor_lambda_star}
        \lambst = \bigg( \Big( \sum\nolimits_{p=1}^{m} (1/n_p) \matr{S}_p \Big)^{-1/2} + \Big(\sum\nolimits_{p=1}^m (1/n_p) \matr{S}_p\Big)^{-1} o_p(n^{-1/2}) \bigg) \vect{z}_n - \vect{\nu}.
    \end{equation}

    Finally, in the expansion of $-2\log R(\bar{\bmu})$, replace the value of $\bar{\bmu}$ to get 
    \begin{equation} \label{eq:cor_expansion}
        2 (\lambst)^\top  \Big(\sum\nolimits_{p=1}^m (1/n_p) \matr{V}_p\Big)^{1/2}\vect{z}_n - (\lambst)^\top \Big(\sum\nolimits_{p=1}^{m} (1/n_p) \matr{S}_p \Big) \lambst + 2 \sum\nolimits_{p=1}^{m} \sum\nolimits_{j=1}^{n_p} \phi_{pj}.
    \end{equation}
    Inserting the value of $\lambst$ obtained in \eqref{eq:cor_lambda_star} into \eqref{eq:cor_expansion} yields
    {
    \setlength{\abovedisplayskip}{5pt}
    \setlength{\belowdisplayskip}{0pt}
    \begin{align*}
        & 2 \bigg( \Big(\sum\nolimits_{p=1}^m (1/n_p) \matr{S}_p\Big)^{-1/2} \vect{z} + \Big(\sum\nolimits_{p=1}^m (1/n_p) \matr{S}_p\Big)^{-1} o_p(n^{-1/2}) \vect{z} \bigg)^\top  \bigg( \Big(\sum\nolimits_{p=1}^m (1/n_p) \matr{S}_p\Big)^{1/2} o_p(n^{-1/2}))\bigg)\vect{z} -  \|\vect{z}\|^2 \\
        & \quad + 2\vect{\nu}^\top \Big(\sum\nolimits_{p=1}^m (1/n_p) \matr{S}_p\Big)^{1/2}\vect{z} 
        - 2 \vect{\nu}^\top \Big(\sum\nolimits_{p=1}^m (1/n_p) \matr{S}_p\Big)^{1/2}\vect{z} - \vect{\nu}^\top \Big(\sum\nolimits_{p=1}^m (1/n_p) \matr{S}_p\Big) \vect{\nu} + 2 \sum\nolimits_{p=1}^{m} \sum\nolimits_{j=1}^{n_p} \phi_{pj} \\
        & = 2 \norm{\vect{z}}^2 - \norm{\vect{z}}^2
        + 2 \vect{z}^\top \Big(\sum\nolimits_{p=1}^m (1/n_p) \matr{S}_p\Big)^{-1/2} o_p(n^{-1/2}) \vect{z} + \vect{z}^\top \Big(\sum\nolimits_{p=1}^m (1/n_p) \matr{S}_p\Big)^{-1} o_p(1/n) \vect{z}\\
        & \quad - \vect{\nu}^\top \Big(\sum\nolimits_{p=1}^m (1/n_p) \matr{S}_p\Big) \vect{\nu} + 2 \sum\nolimits_{p=1}^{m} \sum\nolimits_{j=1}^{n_p} \phi_{pj} 
         = \norm{\vect{z}}^2 + o_p(1),
    \end{align*}
    }
    and the conclusion follows. \hfill 
\end{proof}

\begin{proposition}\label{prop:upper_bound.cov.opt}
    Under the hypotheses of \Cref{thm:multisample-el}, 
    \begin{equation}\label{eq:B-lower-bound}
        \vect{B} \geq \sum\nolimits_{p=1}^{m} \overline{\vect{Y}}_{p} + (\chi^2_{q,1-\alpha})^{1/2}
        \big\{(\vect{e}_\ell^\top \overline{\matr{V}}_n \vect{e}_\ell)^{1/2}\big\}_{\ell \in [q]}
        + o_p(n^{-1/2}) \quad \mbox{as} \quad  n\to\infty.
    \end{equation}
\end{proposition}

\begin{proof}{Proof}
We recall that $C_\alpha = \set{\sum\nolimits_{p=1}^m \sum\nolimits_{j=1}^{n_p} \vect{Y}_{pj} w_{pj} \colon \, \vect{w} \in \mathcal{U}_\alpha }=
    \set{\bmu \in \mathbb{R}^{q} \colon -2 \log(R(\bmu)) \leq \chi^2_{q,1-\alpha}}$. 
    Let $\vect{x}_0^\ell \triangleq \sum\nolimits_{p=1}^{m}\overline{\vect{Y}}_{p} + (\chi^2_{q,1-\alpha})^{1/2} (\overline{\matr{V}}_n \vect{e}_\ell / (\vect{e}_\ell^\top \overline{\matr{V}}_n \vect{e}_\ell )^{-1/2} )$.
    Now, we show that $u_\ell \geq \vect{e}_\ell^\top \vect{x}_0^\ell + o_p(n^{-1/2})$.
    For $0 < \epsilon < (\chi^2_{q,1-\alpha})^{1/2}$,
    let $\vect{x}_{\! \epsilon}^\ell \triangleq \sum\nolimits_{p=1}^{m}\overline{\vect{Y}}_{p} + \big((\chi^2_{q,1-\alpha})^{1/2} - \epsilon \big) \big( \overline{\matr{V}}_n \vect{e}_\ell / (\vect{e}_\ell^\top \overline{\matr{V}}_n \vect{e}_\ell )^{-1/2} \big)$.
    \Cref{cor:thm3} ensures 
    $ -2\log R(\vect{x}_{\! \epsilon}^\ell) \xrightarrow{P} ((\chi^2_{q,1-\alpha})^{1/2} - \epsilon)^2$ as 
    $n\to\infty$. Hence, $\Pr\set{\vect{x}_{\! \epsilon}^\ell \in C_\alpha} \to 1$ as $n \to \infty$,
    which implies $\Pr\set{\vect{x}_{\! \epsilon}^\ell \in C_\alpha, \ell \in [q]} \to 1$. 
    Applying a diagonalization lemma (see, e.g., Corollary~1.18 in \cite{Attouch1984}), we obtain the existence of a sequence
    $\epsilon(n) \in (0,\infty)$ such that $\epsilon(n) \to 0$ and $\Pr\set{\vect{x}_{\! \epsilon(n)}^\ell \in C_\alpha, \ell \in [q]} \to 1$ as $n \to\infty$.
    We have that $\overline{\matr{V}}_n = O(1/n)$ implies $(\overline{\matr{V}}_n \vect{e}_\ell / (\vect{e}_\ell^\top \overline{\matr{V}}_n \vect{e}_\ell )^{-1/2} ) = O(n^{-1/2})$. Combining this fact with $\epsilon(n) = o(1)$,
    we have $\vect{e}_\ell^\top \vect{x}_{\! \epsilon(n)}^\ell= \vect{e}_\ell^\top \vect{x}_0^\ell - \epsilon(n) (\overline{\matr{V}}_n \vect{e}_\ell / (\vect{e}_\ell^\top \overline{\matr{V}}_n \vect{e}_\ell )^{-1/2} ) = \vect{e}_\ell^\top \vect{x}_0^\ell + o(n^{-1/2})$.

    If $\vect{x}_{\! \epsilon(n)}^\ell \in C_\alpha$, then
    $\mathrm{u}_\ell = \max_{\vect{\mu} \in C_\alpha}\, \vect{e}_\ell^\top \vect{\mu} \geq \vect{e}_\ell^\top \vect{x}_{\! \epsilon(n)}^\ell$.
    Therefore, $\vect{x}_{\! \epsilon(n)}^\ell \in C_\alpha$ ensures
    $\max_{\vect{\mu} \in C_\alpha}\, \vect{e}_\ell^\top \vect{\mu} \geq \vect{e}_\ell^\top \vect{x}_0^\ell + o(n^{-1/2})$.
    Since $\Pr\set{\vect{x}_{\! \epsilon(n)}^\ell \in C_\alpha, \ell \in [q]} \to 1$ as $n\to\infty$,
    the probability of the event
    $\{\max_{\vect{\mu} \in C_\alpha}\, \vect{e}_\ell^\top \vect{\mu} \geq \vect{e}_\ell^\top \vect{x}_0^\ell + o(n^{-1/2}), \ell \in [q]\}$
    approaches $1$ as $n\to\infty$.
    Hence, 
    $\Pr\set{\mathrm{u}_\ell 
    \geq \vect{e}_\ell^\top \vect{x}_0^\ell + o(n^{-1/2}), \ell \in [q]} \to 1$ as $n \to\infty$.
    \hfill 
\end{proof}

\section{Proof for \Cref{sec:model}}

\begin{proof}{Proof of \Cref{thm:worst-case-expansion-error}}
For a fixed Solution $i$,
\Cref{thm:worst-case-expansion-error} differs from Proposition~3 in \cite{lam2017optimization} only in the choice of chi-square quantile for the ambiguity set in \eqref{eq:uncertaintyset}: we use $\chi^2_{k-1,1-\alpha}$ instead of $\chi^2_{1,1-\alpha}$. 
This change affects only a constant and therefore does not alter the asymptotic behavior.

We observe that Proposition~3 in \cite{lam2017optimization} states a requirement that the eighth moment of the simulation output function be bounded. However, inspection of their proof reveals that only the second moment is actually utilized. This condition is already satisfied under \Cref{assump:input.length} in the present work.
\hfill
\end{proof}

\section{Proofs for \Cref{sec:optimization,sec:as-algorithm}}

Prior to proving \Cref{prop:algorithm-error,prop:Uhat-coverage,thm:as-asymp-coverage}, we establish several supporting lemmas and propositions.
The following lemma is a minor modification of Lemmas~EC.1 and EC.2 in \cite{lam2017optimization}.
\begin{lemma}\label{lemma:weights-bound}
Let 
$0 < l(\alpha) < 1 < u(\alpha) < \infty$ solve $x  \exp{(1 + (1/2)\chi^2_{k-1,1-\alpha} - x)} = 1$. Then each $\bw \in \mathcal{U}_\alpha$ satisfies 
    \begin{equation}
    \label{eq:lemma-lq}
        l(\alpha)/n_p \leq w_{pj} \leq u(\alpha)/n_p, \forall j \in [n_p], p \in [m],
    \quad \text{and} \quad  \sum\nolimits_{p=1}^m n_p^2 \sum\nolimits_{j=1}^{n_p} (w_{pj} - 1/n_p)^2 \leq u(\alpha)^2\chi^2_{k-1,1-\alpha}.
    \end{equation} 
\end{lemma}

\begin{proof}{Proof}
We omit the proof  as it follows directly from the proofs of Lemmas~EC.1 and EC.2 in \cite{lam2017optimization}, with the sole distinction that our ambiguity set is defined using $\chi^2_{k-1,1-\alpha}$ instead of $\chi^2_{1,1-\alpha}$.
\hfill
\end{proof}

\begin{proposition}\label{prop:algorithm-error-r2}
    Under \Cref{assump:data_size,assump:input.length,assump:pos.def.cov,assump:sim_output_bounded}, for $i\neq \ell$,
    $
        \big|\widehat{U}_{i\ell,n} - \big(\eta_i(\widehat{\bw}^*_{i\ell}) - \eta_\ell(\widehat{\bw}^*_{i\ell})\big)\big| = O_p\big(R_2^{-1/2}\big).
    $
\end{proposition}
\begin{proof}{Proof}
    \setlength{\abovedisplayskip}{5pt}
    \setlength{\belowdisplayskip}{0pt} 
    For any constant $M>0$, Chebyshev's inequality ensures
\begin{align}
        \Pr\Big( \big| \widehat{U}_{i\ell,n} - (\eta_i(\widehat{\bw}^*_{i\ell}) - \eta_\ell(\widehat{\bw}^*_{i\ell})) \big| > MR \Big) 
        & \leq R_2 M^{-2}\E\Big[\E\Big[ \big| \widehat{U}_{i\ell,n} - (\eta_i(\widehat{\bw}^*_{i\ell}) - \eta_\ell(\widehat{\bw}^*_{i\ell})) \big|^2 \, \Big| \, \widehat{\bw}^*_{i\ell} \Big] \Big]. \label{eq:second.moment.bound} 
    \end{align}
The definition of $\widehat{U}_{i\ell,n}$ in~\eqref{eq:U_iell.estimate} yields $\E[ | \widehat{U}_{i\ell,n} - (\eta_i(\widehat{\bw}^*_{i\ell}) - \eta_\ell(\widehat{\bw}^*_{i\ell})) |^2 \, | \, \widehat{\bw}^*_{i\ell} ] = R_2^{-1}\V\big(h_i(\matr{Z}^{i}_{1},\ldots,\matr{Z}^{i}_{m}) - h_\ell(\matr{Z}^{\ell}_{1},\ldots,\matr{Z}^{\ell}_{m}) |\widehat{\bw}^*_{i\ell}\big)$. Since  
    \Cref{assump:sim_output_bounded} bounds $\V_{\bw} [h_i(\vect{X}_{i1},\ldots,\vect{X}_{im}) - h_\ell(\vect{X}_{\ell 1}, \ldots,\vect{X}_{\ell m})]$  uniformly over all $\bw\in\mathcal{U}_\alpha$, the right-hand side in \eqref{eq:second.moment.bound} does not exceed $\mathcal{M}/M^2$ for some constant $\mathcal{M} >0$. 
    Therefore, for any $\varepsilon > 0$, choosing $M=\sqrt{\mathcal{M}/\varepsilon}$,
    we have $\Pr(\big|\widehat{U}_{i\ell,n} - (\eta_i(\widehat{\bw}^*_{i\ell}) - \eta_\ell(\widehat{\bw}^*_{i\ell}))\big| > MR_2^{-1/2}) \leq \varepsilon$.\hfill
\end{proof}

\begin{proposition}\label{prop:algorithm-error-sups}
If \Cref{assump:input.length} holds, then
    for $i\neq \ell$, 
     $   \big|\eta_i(\widehat{\bw}^*_{i\ell}) - \eta_\ell(\widehat{\bw}^*_{i\ell}) - \big( \eta_i(\bw^*_{i\ell}) - \eta_\ell(\bw^*_{i\ell}) \big) \big| \leq
        2  \sup_{\bw \in \mathcal{U}_\alpha} \big| \eta_i(\bw) - \eta_\ell(\bw) - \big( \widehat{\eta}^L_i(\bw) - \widehat{\eta}^L_\ell(\bw) \big)  \big|$.
\end{proposition}

\begin{proof}{Proof}
Let us define
$E_{i\ell} \triangleq \sup_{\bw \in \mathcal{U}_\alpha} \big| \eta_i(\bw) - \eta_\ell(\bw) - \big( \widehat{\eta}^L_i(\bw) - \widehat{\eta}^L_\ell(\bw) \big)  \big|$.
    Using the triangle inequality, $ \big|\eta_i(\widehat{\bw}^*_{i\ell}) - \eta_\ell(\widehat{\bw}^*_{i\ell}) - \big( \eta_i(\bw^*_{i\ell}) - \eta_\ell(\bw^*_{i\ell}) \big) \big|$ can be bounded from above by
    \begin{equation}\label{eq:prop-error-sups-1}
        \big|\eta_i(\widehat{\bw}^*_{i\ell}) - \eta_\ell(\widehat{\bw}^*_{i\ell}) - 
        \big(\widehat{\eta}^L_i(\widehat{\bw}^*_{i\ell}) - \widehat{\eta}^L_\ell(\widehat{\bw}^*_{i\ell})\big) \big| + \big|
        \widehat{\eta}^L_i(\widehat{\bw}^*_{i\ell}) - \widehat{\eta}^L_\ell(\widehat{\bw}^*_{i\ell}) - \big( \eta_i(\bw^*_{i\ell}) - \eta_\ell(\bw^*_{i\ell}) \big) \big|.
    \end{equation}
    The first term in \eqref{eq:prop-error-sups-1} 
    is less than or equal to $E_{i\ell}$.
    For the second term, we observe
    {
    \setlength{\abovedisplayskip}{5pt}
    \setlength{\belowdisplayskip}{0pt}
    \begin{align*}
        \eta_i(\bw^*_{i\ell}) - \eta_\ell(\bw^*_{i\ell})
        & = \sup\nolimits_{\bw \in \mathcal{U}_\alpha} \left\{\widehat{\eta}^L_i(\bw) - \widehat{\eta}^L_\ell(\bw) + \eta_i(\bw) - \eta_\ell(\bw) - \big( \widehat{\eta}^L_i(\bw) - \widehat{\eta}^L_\ell(\bw)\big)\right\}\\
        & \leq  \sup\nolimits_{\bw \in \mathcal{U}_\alpha} \left\{\widehat{\eta}^L_i(\bw) - \widehat{\eta}^L_\ell(\bw)\right\} + E_{i\ell}
         =\widehat{\eta}^L_i(\widehat{\bw}^*_{i\ell}) - \widehat{\eta}^L_\ell(\widehat{\bw}^*_{i\ell}) + E_{i\ell}.
    \end{align*}
    }%
    Similarly, we can show $\widehat{\eta}^L_i(\widehat{\bw}^*_{i\ell}) - \widehat{\eta}^L_\ell(\widehat{\bw}^*_{i\ell}) \leq \eta_i(\bw^*_{i\ell}) - \eta_\ell(\bw^*_{i\ell})  + E_{i\ell}$. Combining the two expressions, we have $\big|\widehat{\eta}^L_i(\widehat{\bw}^*_{i\ell}) - \widehat{\eta}^L_\ell(\widehat{\bw}^*_{i\ell}) - \big( \eta_i(\bw^*_{i\ell}) - \eta_\ell(\bw^*_{i\ell}) \big) \big| \leq E_{i\ell}$. 
    \hfill
\end{proof}

\begin{lemma}\label{lemma:if-sums}
If \Cref{assump:input.length} holds, then
for each $i \in [k]$ and $p \in [m]$,  $\sum\nolimits_{j=1}^{n_p} \IFhatXipj$ \\$= \sum\nolimits_{j=1}^{n_p} \IFhathatXipj = 0$.
\end{lemma}
\begin{proof}{Proof} 
    \Cref{eq:etaidifferentiable}  ensures $(1/n_p)\sum\nolimits_{j=1}^{n_p} \IFhatXipj = \int \IF_{ip}(\bx;\bFhat)\mathrm{d}\Fhat_p(\bx) = 0$.
    Using \eqref{eq:if_estimator}, we note that $\IFhathatXipj$ is written as a sample covariance, where the first term does not depend on the index $j$. Therefore, using the linearity of covariance to move the sum inside the second term of the covariance,
    {
    \setlength{\abovedisplayskip}{5pt}
    \setlength{\belowdisplayskip}{0pt}
    \begin{align*}
        \sum\nolimits_{j=1}^{n_p} \IFhathatXipj 
        & = R_1^{-1} \sum\nolimits_{r=1}^{R_1} \Big(h_i(\matr{Z}^{ir}_{1},\ldots,\matr{Z}^{ir}_{m}) - \overline{h}_i\Big) \Big(n_p \sum\nolimits_{t=1}^{T_{ip}} \sum\nolimits_{j=1}^{n_p} \one{\bZ_{p}^{ir}(t) = \bX_{pj}} - n_pT_{ip}\Big)\\
        & = R_1^{-1} \sum\nolimits_{r=1}^{R_1} \Big(h_i(\matr{Z}^{ir}_{1},\ldots,\matr{Z}^{ir}_{m}) - \overline{h}_i\Big) \Big(n_p \sum\nolimits_{t=1}^{T_{ip}} 1 - n_pT_{ip}\Big) = 0. \hfill
    \end{align*}
    }%
\end{proof}

\begin{lemma}\label{lemma:covariance}
    Let $A$, $B$, and $C$ be random variables such that $\E[B] = 0$, and $C$ is independent of $A$ and $B$.
    Then $\Cov(AB,CB) = \E[AB^2]\E[C]$. If $A$ and $B$ are independent, then $\Cov(AB,CB) = \E[A]\E[B^2]\E[C]$.
    Moreover, if $\E[C] = 0$, then $\Cov(AB,AC) = 0$.
\end{lemma}

\begin{proof}{Proof} 
The statements are easily verified.
 \hfill
\end{proof}

\begin{proposition}\label{prop:algorithm-error-r1}
    Under \Cref{assump:data_size,assump:input.length,assump:pos.def.cov,assump:sim_output_bounded},  as $n\to \infty$ and $R_1 \to \infty$,
    $   \sup_{\bw \in \mathcal{U}_\alpha} 
        \big| \eta^L_i(\bw; \bFhat) - \eta^L_\ell(\bw; \bFhat) - \big( \widehat{\eta}^L_i(\bw) - \widehat{\eta}^L_\ell(\bw)\big) \big|
        = O_p(R_1^{-1/2})$, for $i\neq \ell$.
\end{proposition}

\begin{proof}{Proof} 
    We have
    $
        \sup_{\bw \in \mathcal{U}_\alpha} 
        \big| \eta^L_i(\bw; \bFhat) - \eta^L_\ell(\bw; \bFhat) - \big( \widehat{\eta}^L_i(\bw) - \widehat{\eta}^L_\ell(\bw)\big) \big|
        =\sup_{\bw \in \mathcal{U}_\alpha} 
        \big| \sum\nolimits_{p=1}^m \sum\nolimits_{j=1}^{n_p} \big(\IFhatXipj - \IFhatXlpj - \big( \IFhathatXipj - \IFhathatXlpj \big) \big) w_{pj} \big|
    $.
    By the triangle inequality, it suffices to show that $\sup_{\bw \in \mathcal{U}_\alpha}  \big| \sum\nolimits_{p=1}^m \sum\nolimits_{j=1}^{n_p} \big(\IFhatXipj - \IFhathatXipj  \big) w_{pj} \big| = O_p(R_1^{-1/2})$. Using \Cref{lemma:if-sums,lemma:weights-bound}, and the Cauchy--Schwartz inequality, 
    we have
    \setlength{\abovedisplayskip}{5pt}
    \setlength{\belowdisplayskip}{0pt}
    \begin{align}
        &\sup_{\bw \in \mathcal{U}_\alpha} \Big| \sum\nolimits_{p=1}^m \sum\nolimits_{j=1}^{n_p} \big(\IFhatXipj - \IFhathatXipj  \big) w_{pj} \Big|\notag\\
        & = \sup_{\bw \in \mathcal{U}_\alpha} \Big| \sum\nolimits_{p=1}^m \sum\nolimits_{j=1}^{n_p} n_p^{-1}\big(\IFhatXipj - \IFhathatXipj  \big) n_p (w_{pj} - 1/n_p) \Big| \notag\\
        & \leq \sup_{\bw \in \mathcal{U}_\alpha} \Big( \sum\nolimits_{p=1}^m  \sum\nolimits_{j=1}^{n_p} n_p^{-2} \big(\IFhatXipj - \IFhathatXipj  \big)^2 \sum\nolimits_{p=1}^m \sum\nolimits_{j=1}^{n_p} n_p^2 (w_{pj} - 1/n_p)^2 \Big)^{1/2}\notag\\
        & \leq u(\alpha) \big(\chi^2_{k-1,1-\alpha}\big)^{1/2} \Big( \sum\nolimits_{p=1}^m n_p^{-2} \sum\nolimits_{j=1}^{n_p} \big(\IFhatXipj - \IFhathatXipj  \big)^2\Big)^{1/2}. \label{eq:r1-error-final-bound}
    \end{align}
    
    We proceed to show  
    $\E\Big[ n_p^{-2} \sum\nolimits_{j=1}^{n_p} \big(\IFhatXipj - \IFhathatXipj  \big)^2 \Big] = O(R_1^{-1})$.
    For any $p \in [m]$, let $\bXdata \triangleq \{\bX_{pl}\}_{l\in[n_p]}$ be the collection of observations of the $p$th input distribution. Then, noting that the weighted sum over $j$ with equal weights is the same as the expectation over $\bX_{pj} \in \bXdata$ given $\bXdata$,
    and applying the law of total variance,
    {
    \setlength{\abovedisplayskip}{5pt}
    \setlength{\belowdisplayskip}{0pt}
    \begin{align}
         & 
         \E\Big[ n_p^{-2} \sum\nolimits_{j=1}^{n_p} \big(\IFhatXipj - \IFhathatXipj  \big)^2 \Big]
         = \E\bigg[ \E \Big[ \left.n_p^{-1} \big(\IFhatXipj - \IFhathatXipj  \big)^2  \,\right|\, \bXdata \Big] \bigg]\notag\\
         & = n_p^{-1} \Big\{ \E\big[ \V (\IFhathatXipj \,|\,\bXdata) \big] + \E\big[ \E [\IFhathatXipj - \IFhatXipj \,|\, \bXdata]^2 \big] \Big\}.\label{eq:exp2ifhathat}
    \end{align}
    }%
    In the following, we derive upper bounds for each term in~\eqref{eq:exp2ifhathat}.

    To derive the conditional variance in \eqref{eq:exp2ifhathat}, let $\matr{Z}^i \triangleq (\matr{Z}^i_1,\ldots,\matr{Z}^i_m)$ be the collection of random inputs generated to run the $i$th solution's simulation. When we desire to specify the replication  $r\in[R_1]$,  we adopt  $\matr{Z}^{ir} \triangleq (\matr{Z}^{ir}_1,\ldots,\matr{Z}^{ir}_m)$. 
    In the following, let us fix $i\in[k], p\in[m]$ and $j\in[n_p]$ and rewrite $\IFhathatXipj = R_1^{-1} \sum\nolimits_{r=1}^{R_1} (A_r-\bar{A})B_r$, where $A_r \triangleq h_i(\matr{Z}^{ir})$, $\bar{A} \triangleq R_1^{-1} \sum\nolimits_{r=1}^{R_1} A_r$ and $B_r \triangleq n_p \sum\nolimits_{t=1}^{T_{ip}} \one{\bZ_{p}^{ir}(t) = \bX_{pj}} - T_{ip}$.
    In addition, we define $\Etilde[\cdot] \triangleq \E[\cdot | \bXdata]$, $\Vtilde(\cdot) \triangleq \V(\cdot | \bXdata)$ and $\Covtilde(\cdot) \triangleq \Cov(\cdot | \bXdata)$.
    Since $\Etilde[\one{\bZ^i_p(t)=\bX_{pj}}] = 1/n_p$, we have $\Etilde[B_r] = 0$. Furthermore,
    {
    \setlength{\abovedisplayskip}{5pt}
    \setlength{\belowdisplayskip}{0pt}
    \begin{align}
        \V(\IFhathatXipj|\bXdata) & = \Vtilde\left(R_1^{-1} \sum\nolimits_{r=1}^{R_1} (A_r-\bar{A})B_r\right)\notag\\
        & = R_1^{-2} \bigg[ \sum\nolimits_{r=1}^{R_1} \Vtilde((A_r-\bar{A})B_r) + \sum\nolimits_{r\neq s} \Covtilde((A_r-\bar{A})B_r,(A_s-\bar{A})B_s)  \bigg]\notag\\
        & = R_1^{-1} \big[ \Vtilde((A_r-\bar{A})B_r) + (R_1-1) \Covtilde((A_r-\bar{A})B_r,(A_s-\bar{A})B_s)  \big].\label{eq:varifhathat1}
    \end{align}
    }%
    Applying \Cref{lemma:covariance}, the variance term in \eqref{eq:varifhathat1} can be expanded as:
    {
    \footnotesize
    \setlength{\abovedisplayskip}{5pt}
    \setlength{\belowdisplayskip}{0pt}
    \begin{align}
        \Vtilde\big((A_r-\bar{A})B_r\big) 
        &  = \Vtilde\bigg( \frac{R_1-1}{R_1}A_rB_r - \Big(\frac{1}{R_1}\sum\nolimits_{s: s \neq r}A_s\Big)B_r \bigg)\notag\\
        &  = \frac{(R_1-1)^2}{R_1^2}\Vtilde\big(A_rB_r \big)
        - 2\frac{R_1-1}{R_1}\Covtilde\Big(A_rB_r,\Big[\frac{1}{R_1}\sum\nolimits_{s: s \neq r}A_s\Big]B_r\Big)
        + \frac{1}{R_1^2}\Vtilde\Big(\sum\nolimits_{s: s \neq r}A_sB_r\Big)\notag\\
        &  = \frac{(R_1-1)^2}{R_1^2}\Etilde\big[A_r^2B_r^2 \big]
        - \frac{(R_1-1)^2}{R_1^2}\Etilde\big[A_rB_r \big]^2
        - 2\frac{(R_1-1)^2}{R_1^2} \Etilde\big[A_rB_r^2 \big]\Etilde\big[A_s \big] 
        \notag\\& \quad
        + \frac{R_1-1}{R_1^2} \Etilde[A_r^2]\Etilde[B_r^2] + \frac{(R_1-1)(R_1-2)}{R_1^2}\Etilde[A_r]^2\Etilde[B_r^2].\label{eq:varifhathat-var}
    \end{align}
    }%
    We note that, if $r$ and $s$ represent different replications of the simulation which are run independently, then $A_rB_r$ is independent from $A_sB_s$. Using this combined with \Cref{lemma:covariance}, 
    the covariance term in \eqref{eq:varifhathat1} can be expanded as
    {
    \setlength{\abovedisplayskip}{5pt}
    \setlength{\belowdisplayskip}{0pt}
    \begin{align}
        & \Covtilde\big((A_r-\bar{A})B_r,(A_s-\bar{A})B_s\big) 
        =
        \Covtilde(A_rB_r,A_sB_s) - 2\Covtilde(A_rB_r, \bar{A}B_s) + \Covtilde(\bar{A}B_r, \bar{A}B_s) \notag\\
        & \quad  = 0 - 2R_1^{-1} \cdot \Covtilde(A_rB_r, A_rB_s) + \Covtilde\Big(R_1^{-1}A_rB_r + R_1^{-1} \sum\nolimits_{l \neq r} A_lB_r, R_1^{-1}A_sB_s + R_1^{-1} \sum\nolimits_{l \neq s} A_lB_s\Big)\notag\\
        & \quad  = R_1^{-2}\Covtilde(A_rB_r,A_sB_s) + 2R_1^{-2} \Covtilde\Big(A_rB_r,\sum_{l \neq s} A_lB_s\Big) + R_1^{-2} \Covtilde\Big( \sum_{l \neq r} A_lB_r, \sum_{l \neq s} A_lB_s\Big)
        = R_1^{-2} \Etilde[A_rB_r]^2. \label{eq:varifhathat-cov} 
    \end{align}
    }%

    Combining \eqref{eq:varifhathat-var} and \eqref{eq:varifhathat-cov}, \eqref{eq:varifhathat1} can be rewritten and bounded as
    \begin{footnotesize}
    {
    \setlength{\abovedisplayskip}{5pt}
    \setlength{\belowdisplayskip}{0pt}
    \begin{align}
        & R_1^{-1} \bigg\{ \frac{(R_1-1)^2}{R_1^2} \big(\Etilde\big[A_r^2B_r^2 \big]
        -2 \Etilde\big[A_rB_r^2 \big]\Etilde\big[A_s \big] \big)
        + \frac{R_1-1}{R_1^2} \Etilde[A_r^2]\Etilde[B_r^2] + \frac{(R_1-1)(R_1-2)}{R_1^2}(\Etilde[A_r]^2\Etilde[B_r^2] 
        - \Etilde[A_rB_r]^2)
        \bigg\} \notag\\
        &\leq 
        R_1^{-1} \bigg\{\Etilde\big[A_r^2B_r^2 \big]
        - 2\frac{(R_1-1)^2}{R_1^2} \Etilde\big[A_rB_r^2 \big]\Etilde\big[A_r \big]  +\Etilde[A_r]^2\Etilde[B_r^2] 
        + R_1^{-1} 
        \Etilde[A_r^2]\Etilde[B_r^2]
        \bigg\}. \label{eq:varifhathat2}
    \end{align}
    }%
    \end{footnotesize}%

    We can expand each expectation of \eqref{eq:varifhathat2} by plugging back the values of $A_r$ and $B_r$. For this purpose, let 
    us write $h_i$ as an abbreviation of $h_i(\matr{Z}_i)$ for notational convenience.
    We have 
    \begin{footnotesize}
    {
    \setlength{\abovedisplayskip}{5pt}
    \setlength{\belowdisplayskip}{0pt}
    \begin{align*}
        \Etilde[B_r^2] &
        = n_p^2 \sum\nolimits_{t=1}^{T_{ip}} \Vtilde\big(\one{\bZ_{p}^{ir}(t) = \bX_{pj}}\big) + \sum\nolimits_{s \neq t} \Covtilde\big(\one{\bZ_{p}^{ir}(t) = \bX_{pj}},\one{\bZ_{p}^{ir}(s) = \bX_{pj}} \big)\\
        & = n_p^2 T_{ip} n_p^{-1}(1 - n_p^{-1}) + 0 = T_{ip}(n_p - 1),\\
        \Etilde[A_r^qB_r^2] & = \Etilde\bigg[h_i^q \bigg(\Big(n_p\sum\nolimits_{t=1}^{T_{ip}}\one{\bZ^{ir}_p(t)=\bX_{pj}}\Big)^2 - 2n_pT_{ip} \sum\nolimits_{t=1}^{T_{ip}}\one{\bZ^{ir}_p(t)=\bX_{pj}} + T_{ip}^2\bigg)\bigg]\\
        & = n_p^2 \sum\nolimits_{t=1}^{T_{ip}}\Etilde\big[h_i^q \one{\bZ^{ir}_p(t)=\bX_{pj}}\big]
        + n_p^2 \sum\nolimits_{t\neq s}\Etilde\big[h_i^q \one{\bZ^{ir}_p(t)=\bX_{pj}}\one{\bZ^{ir}_p(s)=\bX_{pj}}\big]
        \\ & \quad
        -2 n_pT_{ip} \sum\nolimits_{t=1}^{T_{ip}}\Etilde\big[h_i^q \one{\bZ^{ir}_p(t)=\bX_{pj}}\big]
        + T_{ip}^2 \Etilde[h_i^q]\\
        & = (n_p - 2T_{ip})\sum\nolimits_{t=1}^{T_{ip}} \E\big[h_i^q | \bZ^{ir}_p(t)=\bX_{pj},\bXdata\big] + \sum\nolimits_{t\neq s} \E\big[h_i^q | \bZ^{ir}_p(t)=\bZ^{ir}_p(s)=\bX_{pj}, \bXdata\big] + T_{ip}^2\E[h_i^q | \bXdata],
    \end{align*}
    }%
    \end{footnotesize}%
    for $q \in \{1, 2\}$.   
    Since these expectations are conditional on the data, they are taken with respect to the edf. Thus,
    \begin{equation*}
        \E[h_i^q\,|\, \bXdata] = \prod\nolimits_{p=1}^m (n_p^{-T_{ip}}) \sum\nolimits_{I^i,\ldots,I^m} h(\bbZ^i_{1,I^1},\ldots,\bbZ^i_{m,I^m})^q.
    \end{equation*}
    Under \Cref{assump:sim_output_bounded}, $\E[\E[h_i^q\,|\, \bXdata]] < \infty$. Applying the same logic,
    $\E[\E[h_i^q | \bZ^i_p(t) = \bX_{pj}, \bXdata]]<\infty$ and $\E[\E[h_i^q | \bZ^{ir}_p(t)=\bZ^{ir}_p(s)=\bX_{pj}, \bXdata]] < \infty$ for $q=1,2$. Since $T_{ip}$ is a fixed constant and taking the outer expectations, we find that
    $\E\big[\Etilde[A_r^qB_r^2]\big]$ is equal to
    {
    \setlength{\abovedisplayskip}{5pt}
    \setlength{\belowdisplayskip}{0pt}
    \begin{align}
    &(n_p - 2T_{ip})\sum_{t=1}^{T_{ip}} \E\big[\E[h_i^q | \bZ^{ir}_p(t)=\bX_{pj},\bXdata] \big] + \sum_{t\neq s} \E\big[\E[h_i^q | \bZ^{ir}_p(t)=\bZ^{ir}_p(s)=\bX_{pj}, \bXdata]\big] + T_{ip}^2\E\big[\E[h_i^q| \bXdata]\big] \notag\\
    & = O(n_p) + O(1) + O(1) = O(n_p).\label{eq:E.of.E[AB2]}
    \end{align}
    Analogously, 
    \begin{align}
    \E\big[\Etilde[A_rB_r]^2\big] 
    & = 
    \sum_{t=1}^{T_{ip}} \E\big[\E[h_i|\bZ^{ir}_p(t) = \bX_{pj},\bXdata]^2\big] + \sum_{t \neq s} \E\big[\E[h_i|\bZ^{ir}_p(t) = \bX_{pj},\bXdata]\E[h_i|\bZ^{ir}_p(s) = \bX_{pj},\bXdata]\big] 
    \nonumber\\ & \quad
    - 2T_{ip} \sum\nolimits_{t=1}^{T_{ip}} \E\big[ \E[h_i|\bZ^{ir}_p(t) = \bX_{pj},\bXdata]\E[h_i|\bXdata]\big] + T_{ip}^2\E\big[\E[h_i|\bXdata]^2\big]
     =
    O(1). \label{eq:E.of.E[AB]^2}
    \end{align} 
    }%
    Finally, combining that $
    \E[\Etilde[B_r^2]] = \E[T_{ip}(n_p-1)] = O(n_p)$ with \eqref{eq:varifhathat2}, \eqref{eq:E.of.E[AB2]} and \eqref{eq:E.of.E[AB]^2}, we bound the unconditioned expectation of the variance of $\IFhathatXipj$ as
    {
    \setlength{\abovedisplayskip}{5pt}
    \setlength{\belowdisplayskip}{0pt}
    \begin{align*}
     \E\big[\Vtilde(\IFhathatXipj)\big]
        & \leq \E\bigg[R_1^{-1} \bigg\{\Etilde\big[A_r^2B_r^2 \big]
        - 2\frac{(R_1-1)^2}{R_1^2} \Etilde\big[A_rB_r^2 \big]\Etilde\big[A_r \big]  +\Etilde[A_r]^2\Etilde[B_r^2] 
        + R_1^{-1} 
        \Etilde[A_r^2]\Etilde[B_r^2]
        \bigg\}\bigg] \\
        & \leq R_1^{-1} \bigg\{
        O(n_p)
        - 2\frac{(R_1-1)^2}{R_1^2} O(n_p)  + O(n_p) 
        + R_1^{-1} \Big\{ 
        O(n_p) 
        + O(1) \Big\}
        \bigg\} 
        = O(n_p R_1^{-1}).
    \end{align*}
    }%
    Now, we tackle the squared bias term in \eqref{eq:exp2ifhathat}. Using the independence of $r$ from $r'$, and that $\Etilde[B_r]=0$, we have
    \begin{equation*}
        \E[\IFhathatXipj \,|\, \bXdata]
        = \Etilde[(A_r-\bar{A})B_r] 
        = \frac{R_1-1}{R_1}\Etilde[A_rB_r] + \frac{R_1-1}{R_1}\Etilde[A_{r'}B_r]
        = \frac{R_1-1}{R_1}\Etilde[A_rB_r].
    \end{equation*} 
    Therefore, using \eqref{eq:E.of.E[AB]^2}, $\E\big[ \E [\IFhathatXipj - \IFhatXipj \,|\, \bXdata]^2\big] = \E\big[\Etilde[A_rB_r]^2\big]/R_1^2 = O(R_1^{-2})$.

    Finally, combining all pieces,
    {
    \setlength{\abovedisplayskip}{5pt}
    \setlength{\belowdisplayskip}{0pt}
    \begin{align*}
        \eqref{eq:exp2ifhathat} & = n_p^{-1} \Big\{ \E\Big[ \V (\IFhathatXipj \,|\,\bXdata) \Big] + \E\Big[ \E [\IFhathatXipj - \IFhatXipj \,|\, \bXdata]^2 \Big] \Big\}\\
        & = n_p^{-1} \big\{ O(n_pR_1^{-1}) + O(R_1^{-2}) \big\}= O(R_1^{-1}).
    \end{align*}
    }%
    Hence, the right-hand side in
    \eqref{eq:r1-error-final-bound} is $O_p(R_1^{-1/2})$.
    \hfill
\end{proof}

\begin{proof}{Proof of \Cref{prop:algorithm-error}}
    By the triangle inequality, 
    
    \begin{footnotesize}
    \begin{equation}\label{eq:prop-error-pf-1}
        \big|\widehat{U}_{i\ell,n} - U_{i\ell,n} \big| \! \leq \!\big|\widehat{U}_{i\ell,n} - (\eta_i(\widehat{\bw}^*_{i\ell}) - \eta_\ell(\widehat{\bw}^*_{i\ell}))\big| \! + \!\big|\eta_i(\widehat{\bw}^*_{i\ell}) - \eta_\ell(\widehat{\bw}^*_{i\ell}) - \big( \eta_i(\bw^*_{i\ell}) - \eta_\ell(\bw^*_{i\ell}) \big) \big| \! + \! \big| \eta_i(\bw^*_{i\ell}) - \eta_\ell(\bw^*_{i\ell}) - U_{i\ell,n} \big|.
    \end{equation}
    \end{footnotesize}%
    From \Cref{prop:algorithm-error-r2}, the first term in the upper bound of \eqref{eq:prop-error-pf-1} is $O_p(R_2^{-1/2})$. For the second term, from \Cref{prop:algorithm-error-sups}, $\big|\eta_i(\widehat{\bw}^*_{i\ell}) - \eta_\ell(\widehat{\bw}^*_{i\ell}) - \big( \eta_i(\bw^*_{i\ell}) - \eta_\ell(\bw^*_{i\ell}) \big) \big| \leq 2 \sup_{\bw \in \mathcal{U}_\alpha} \big| \eta_i(\bw) - \eta_\ell(\bw) - \big( \widehat{\eta}^L_i(\bw) - \widehat{\eta}^L_\ell(\bw)\big) \big|$, where the supremum can be bounded above by the sum of two terms by the triangle inequality: (i) $ \sup_{\bw \in \mathcal{U}_\alpha} \big| \eta_i(\bw) - \eta_\ell(\bw) - \big( \eta^L_i(\bw; \bFhat) - \eta^L_\ell(\bw; \bFhat)\big) \big|$; and (ii) $\sup_{\bw \in \mathcal{U}_\alpha} \big| \eta^L_i(\bw; \bFhat) - \eta^L_\ell(\bw; \bFhat) - \big( \widehat{\eta}^L_i(\bw) - \widehat{\eta}^L_\ell(\bw)\big) \big|.$
   Part~(i) can be bounded from \eqref{eq:worst-case-error} and \Cref{prop:algorithm-error-r1} states that (ii) shrinks at $O_p(R_1^{-1/2})$. 
   Lastly, by Theorem~2.2 in \cite{kaniovski1995probabilistic}, the last term of \eqref{eq:prop-error-pf-1} can be bounded with $\sup_{\bw \in \mathcal{U}_\alpha} \big|\eta_i(\bw) - \eta_\ell(\bw) - \big( \eta^L_i(\bw) - \eta^L_\ell(\bw) \big) \big|$. Combining it with \Cref{thm:worst-case-expansion-error}, we get $\big|\eta_i(\bw^*_{i\ell}) - \eta_\ell(\bw^*_{i\ell}) - U_{i\ell,n} \big|
    = O_p(1/n)$. 
   Combining all pieces, 
   $\big|\widehat{U}_{i\ell,n} - U_{i\ell,n} \big| = o_p(n^{-1/2}) + O_p(R_1^{-1/2}) + O_p(R_2^{-1/2})$.
   \hfill
\end{proof}

To prove \Cref{prop:Uhat-coverage}, we need first to provide \Cref{lem:prob_chisquare} and \Cref{prop:normality_coverage}.
First, to make an analogy to~\eqref{eq:B-lower-bound}, let $\overline{\IF}_{ip} \triangleq n_p^{-1}\sum_{j=1}^{n_p} \IFXipj$ and let
\begin{equation}
\label{eq:overlineUielln}
\overline{U}_{i\ell,n} \triangleq \eta_i(\bF^c) - \eta_\ell(\bF^c) + \sum\nolimits_{p=1}^{m}(\overline{\IF}_{ip} - \overline{\IF}_{\ell p}) + \Big(\chi^2_{k-1,1-\alpha} \sum\nolimits_{p=1}^m n_p^{-1}\V\big(\IFipx - \IFlpx\big)\Big)^{1/2},
\end{equation}
which is obtained by rewriting the right-hand side of~\eqref{eq:B-lower-bound} with $\vect{Y}_{pj}$ in~\eqref{eq:redefined.Y}.

\begin{lemma}\label{lem:prob_chisquare}
     If $\vect{Z} \in \mathbb{R}^q \sim \N(\vect{0},\matr{V})$, and  $a_\ell \triangleq (\chi^2_{q, 1-\alpha} V_{\ell \ell})^{1/2}$, then $\Pr\set{\vect{Z} \leq \vect{a}} \geq 1 - \alpha$.
\end{lemma}
\begin{proof}{Proof}
    We have
    $\Pr\set{\vect{Z}^\top \matr{V} \vect{Z} \leq \chi^2_{q,1-\alpha}} = 1 - \alpha$, and
    $a_\ell = \max \{\vect{e}_\ell^\top \vect{z} \colon  \vect{z}^\top \matr{V} \vect{z} \leq \chi^2_{q,1-\alpha}\}$. It follows that if $\vect{Z} \in \set{\vect{z}^\top \matr{V} \vect{z} \leq \chi^2_{q,1-\alpha}}$, then $\vect{e}_\ell^\top\vect{Z} \leq a_\ell$, 
    $\ell = 1, \ldots, q$.
    Hence
    $\set{\vect{Z}^\top \matr{V} \vect{Z} \leq \chi^2_{q,1-\alpha}} \subseteq \set{\vect{Z} \leq \vect{a}}$. \hfill
\end{proof}

\begin{proposition}\label{prop:normality_coverage}
Under \Cref{assump:input.length,assump:data_size,assump:pos.def.cov}, for each $i \in [k]$,
$\liminf_{n\to \infty}\Pr\set{\eta_i(\bF^c) - \eta_\ell(\bF^c) \leq \overline{U}_{i\ell,n} + o_p(n^{-1/2}), \forall  \ell \neq i} \geq 1 - \alpha$.
\end{proposition}

\begin{proof}{Proof}
Using the definition of $\overline{U}_{i\ell,n}$ provided in \eqref{eq:overlineUielln},
    {\footnotesize
    \setlength{\abovedisplayskip}{5pt}
    \setlength{\belowdisplayskip}{0pt}
    \begin{align}
    & \Pr\set{\eta_i(\bF^c) - \eta_\ell(\bF^c) \leq \overline{U}_{i\ell,n} + o_p(n^{-1/2}), \forall   \ell \neq i} 
    \notag\\
    & \quad = \Pr\Big\{\sum\nolimits_{p=1}^m \big(\overline{\IF}_{ip} - \overline{\IF}_{\ell p}\big) + \Big(\chi^2_{k-1,1-\alpha} \sum\nolimits_{p=1}^m n_p^{-1}\V(\IFipx - \IFlpx)\Big)^{1/2} + o_p(n^{-1/2}) \geq 0, \forall \ell \neq i \Big\} \notag\\
    & \quad = \Pr\Big\{n^{1/2}\sum\nolimits_{p=1}^m({\overline{\IF}_{\ell p} - \overline{\IF}_{ip}}) + o_p(1) \leq  \Big( \chi^2_{k-1,1-\alpha} \sum\nolimits_{p=1}^m \frac{n}{n_p}\V\big(\IFipx - \IFlpx\big)\Big)^{1/2}, \forall \ell \neq i\Big\}. \label{eq:propproof_clt}
    \end{align}   
    }%
By the CLT and Slutsky's theorem, $n^{1/2}\big\{\sum\nolimits_{p=1}^m({\overline{\IF}_{\ell p} - \overline{\IF}_{ip}})\big\}_{\ell=1,\dots,k-1} + o_p(1) \xrightarrow{D} \vect{\mathrm{Z}} \sim \N(0,\matr{\Sigma}^i)$ as $n\to\infty$, where
\[
\Sigma^i_{\ell\ell'} \triangleq \sum\nolimits_{p=1}^m \beta_p^{-1} \E \Big[ \big( \IFipx - \IFlpx \big) \big( \IFipx - \IFllpx \big) \Big]. 
\]
In addition, by \Cref{lem:prob_chisquare}, $\Pr\{\vect{\mathrm{Z}} \leq \vect{a}\} \geq 1-\alpha$, where $\vect{a}= \big\{(\chi^2_{k-1,1-\alpha})^{1/2} \big[ \sum\nolimits_{p=1}^m \beta_p^{-1} \V\big(\IFipx - \IFlpx\big)\big]^{1/2} \big\}_{\ell \neq i}$. Therefore, since the convergence in distribution implies the pointwise convergence of the cumulative distributions functions, 
the probabilities in \eqref{eq:propproof_clt} converge to $\Pr\{\vect{\mathrm{Z}} \leq \vect{a}\}$ as $n\to\infty$. \hfill
\end{proof}

\begin{proof}{Proof of \Cref{prop:Uhat-coverage}}
    From \Cref{prop:algorithm-error}, for each $i \neq \ell$, $\widehat{U}_{i\ell,n} = U_{i\ell,n} + o_p(n^{-1/2}) + O_p(R_1^{-1/2} + R_2^{-1/2})$. Hence $n$, $R_1/n$ and $R_2/n \to\infty$ ensure $\widehat{U}_{i\ell,n} = U_{i\ell,n} + o_p(n^{-1/2})$. Applying \Cref{prop:upper_bound.cov.opt} yields $U_{i\ell,n} \geq \overline{U}_{i\ell,n} + o_p(n^{-1/2})$. Hence $\widehat{U}_{i\ell,n} \geq \overline{U}_{i\ell,n} + o_p(n^{-1/2})$. Therefore, as $n$, $R_1/n$ and $R_2/n \to\infty$, we have
    \begin{align*}
        \Pr\set{\eta_i(\bF^c) - \eta_\ell(\bF^c) \leq \widehat{U}_{i\ell,n},\ell \neq i} 
        & \geq \Pr\set{\eta_i(\bF^c) - \eta_\ell(\bF^c) \leq \overline{U}_{i\ell,n} + o_p(n^{-1/2}),\ell \neq i}
        \geq 1 - \alpha. \hfill 
    \end{align*}
\end{proof}%

\begin{proof}{Proof of \Cref{thm:as-asymp-coverage}}
    \Cref{prop:Uhat-coverage} ensures $ \liminf_{n\to\infty, \, R_1/n\to\infty, \, R_2/n\to\infty} \Pr\set{\eta_i(\bF^c) - \eta_\ell(\bF^c) \leq \widehat{U}_{i\ell,n},\ell \neq i} \geq 1-\alpha$.  Therefore, the assertion follows from an application of \Cref{lem:asymptotic-mcb-ci}.
    \hfill
\end{proof}

\end{document}